\documentclass[orivec]{llncs}
\usepackage{amsmath}
\usepackage{amssymb}
\usepackage{amsfonts,enumitem}
\usepackage[misc,geometry]{ifsym} 
\usepackage{graphicx}
\DeclareGraphicsExtensions{.pdf,.png,.jpg}
\usepackage{tikz}
\usetikzlibrary{petri}
\usetikzlibrary{decorations,decorations.pathreplacing,decorations.pathmorphing,decorations.markings}
\usetikzlibrary{arrows}
\usetikzlibrary{shapes.geometric}
\usetikzlibrary{shapes.symbols}
\usepackage{tkz-graph}
\usepackage[T1]{fontenc} 
\usepackage{mathbbol}
\DeclareSymbolFontAlphabet{\mathbbl}{bbold}

\begin{document}

\newcommand{\lts}{LTS}
\newcommand{\p}[1]{\mbox{{#1}\hspace{-0.20em}\c{}\hspace{0.20em}}}
\newcommand{\anz}{{\tt \#}}\newcommand{\infunnel}{{}^\circ}
\newcommand{\es}{\emptyset}
\newcommand{\pminus}{
\mbox{\textrm{$-$}\!\!\!\!\!\!\!\:\,\,\raisebox{1.5mm}{$\scriptstyle \bullet$}}\:}
\newcommand{\fto}[1]{\stackrel{#1}{\rightarrow}}
\newcommand{\ftov}[1]{\downarrow\!{\scriptstyle #1}}
\newcommand{\dt}{\bullet}
\newcommand{\leer}{\varepsilon}
\newcommand{\nsymbol}{\mathbb{N}}
\newcommand{\zsymbol}{\mathbb{Z}}
\newcommand{\qsymbol}{\mathbb{Q}}
\newcommand{\rsymbol}{\mathbb{R}}
\newcommand{\minus}{\setminus}
\newcommand{\support}{\mathit{supp}}
\newcommand{\goesto}{\rightarrow}
\renewcommand{\goesto}[1]{\stackrel{#1}{\longrightarrow}}
\newcommand{\longgoesto}[1]{\stackrel{#1}{-\!\!\!-\!\!\!-\!\!\!-\!\!\!-\!\!\!-\!\!\!-\!\!\!-\!\!\!-\!\!\!-\!\!\!\!\!\!\longrightarrow}}
\newcommand{\verylonggoesto}[1]{\stackrel{#1}{-\!\!\!-\!\!\!-\!\!\!-\!\!\!-\!\!\!-\!\!\!-\!\!\!-\!\!\!-\!\!\!-\!\!\!-\!\!\!-\!\!\!-\!\!\!-\!\!\!-\!\!\!-\!\!\!-\!\!\!-\!\!\!-\!\!\!-\!\!\!-\!\!\!-\!\!\!-\!\!\!\!\!\!\longrightarrow}}
\newcommand{\ggoesto}[1]{\stackrel{#1}{\Longrightarrow}}
\newcommand{\emptyseq}{\varepsilon}
\newcommand{\impl}{\Rightarrow}
\newcommand{\Parikh}{{\mathbf P}}
\def\tp{^{\sf T}}
\newcommand{\tile}[4]{\begin{array}{|ll|}\hline#1&#2\\#3&#4\\\hline\end{array}}
\newcommand{\disjcup}{\makebox[1.0em][c]{
\setlength{\unitlength}{0.4em}
\begin{picture}(2,2)(0,0)
\put(0.84,0.48){{\tiny $\bullet$}}
\put(0.5,0){$\cup$}
\end{picture}
\setlength{\unitlength}{1mm}
}}

\def\projection#1#2{\mathchoice
              {\setbox1\hbox{${\displaystyle #1}_{\scriptstyle #2}$}
              \projectionaux{#1}{#2}}
              {\setbox1\hbox{${\textstyle #1}_{\scriptstyle #2}$}
              \projectionaux{#1}{#2}}
              {\setbox1\hbox{${\scriptstyle #1}_{\scriptscriptstyle #2}$}
              \projectionaux{#1}{#2}}
              {\setbox1\hbox{${\scriptscriptstyle #1}_{\scriptscriptstyle #2}$}
              \projectionaux{#1}{#2}}}
\def\projectionaux#1#2{{#1\,\smash{\vrule height .8\ht1 depth .85\dp1}}_{\,#2}} 

\newcommand{\choice}{\protect\makebox[1.0em][c]{
\protect\setlength{\unitlength}{0.2em}
\protect\begin{picture}(2,4)(0,0)
\protect\put(0,0){\line(1,0){2}}
\protect\put(0,0){\line(0,0){4}}
\protect\put(0,4){\line(1,0){2}}
\protect\put(2,0){\line(0,1){4}}
\protect\end{picture}
\protect\setlength{\unitlength}{1mm}
}}

\newcommand{\prefix}{\sqsubseteq}
\newcommand{\is}{\iota}
\newcommand{\backwd}[1]{\stackrel{\leftharpoonup}{#1}}
\newcommand{\down}{\;\downarrow\!}
\renewcommand{\vec}[1]{\overrightarrow{#1}}
\newcommand{\uniqueP}{\Upsilon}
\newcommand{\vzero}{0}
\newcommand{\from}{\leftarrow}
\newcommand{\drop}[1]{} 
\newcommand{\one}{\mathbbl{1}}
\newcommand{\zero}{\mathbbl{0}}

\renewcommand{\labelitemi}{$\bullet$}
\renewcommand{\labelitemii}{$\circ$}
\renewcommand{\labelitemiii}{$\cdot$}
\renewcommand{\labelitemiv}{$\ast$}
\newcommand{\TS}{\mathit{TS}}
\newcounter{exTS}
\newcommand{\TSref}[1]{\TS_{\ref{#1}}}
\newcommand{\squeese}{\bot}

\newcommand{\lge}{{\mathcal L}} 
\newcommand{\pref}{{ {PREF}}}

\itemsep0pt

\tikzstyle{mypetristyle}=[
      place/.style={circle, draw=black, fill=white, thick, minimum size=4mm},
      transition/.style={rectangle, draw=black, fill=black!30, thick, minimum height=3mm, minimum width=3mm},
      token/.style={circle,draw=black,fill=black,inner sep=0pt,minimum size=1mm}
]

\tikzset{every picture/.style={mypetristyle}}
\tikzset{every label/.style={black!90}}

\title{Synthesis of Weighted Marked Graphs from Constrained Labelled Transition Systems:\\
A Geometric Approach} 
\titlerunning{Synthesis of WMGs from Constrained LTSs: a Geometric Approach}  
\author{Raymond Devillers\inst{1} \and Evgeny Erofeev\thanks{Supported by DFG through \mbox{grant Be 1267/16-1} {\tt ASYST}.}\inst{2} \and Thomas Hujsa(\Letter)\thanks{Supported by  the STAE foundation/project DAEDALUS, Toulouse, France.}\inst{3}}
\institute{
D\'epartement d'Informatique, Universit\'e Libre de Bruxelles,\\ B-1050 Brussels, Belgium
(\email{rdevil@ulb.ac.be})
\and
Department of Computing Science, Carl von Ossietzky Universit\"at Oldenburg,\\ D-26111 Oldenburg, Germany 
(\email{evgeny.erofeev@uni-oldenburg.de})
\and
LAAS-CNRS, Universit\'e de Toulouse, CNRS, Toulouse, France
(\email{thujsa@laas.fr})
}

\maketitle

\begin{abstract}  
Recent studies investigated the problems of analysing Petri nets and synthesising them
from labelled transition systems (LTS) with two labels (transitions) only. 
In this paper, 
we extend these works by providing new 
conditions
for the synthesis of Weighted Marked Graphs (WMGs),
a well-known and useful class of weighted Petri nets in which each place has at most one input and one output.

Some of these new conditions 
do not restrict the number of labels;
the other ones consider up to 3 labels.
Additional constraints are investigated:
when the LTS is either finite or infinite,
and
either cyclic or acyclic.
We show that one of these conditions, 
developed for 3 labels,
does not extend
to 4 nor to 5 labels.
Also, we tackle geometrically
the WMG-solvability 
of finite, acyclic LTS 
with any number of labels. 
\end{abstract}


\section{Introduction}\label{intro.sec}

Petri nets form a highly expressive and intuitive operational model of discrete event systems,
capturing the mechanisms of synchronisation, conflict and concurrency. 
Many of their fundamental behavioural properties are decidable, 
allowing to model and analyse numerous artificial and natural systems.
However,
most interesting model checking problems are intractable, 
and the efficiency of synthesis algorithms varies widely 
depending 
on the constraints imposed on the desired solution.
In this study,
we focus on the Petri net synthesis problem from a labelled transition system (LTS),
which consists in determining the existence of a Petri net whose reachability graph
is isomorphic to the given LTS,
and building such a Petri net solution when it exists.

In previous studies on the analysis or synthesis of Petri nets,
structural restrictions encompassed \emph{plain} nets (each weight equals $1$; also called ordinary nets)  \cite{murata89},
\emph{homogeneous} nets (meaning that for each place $p$, all the output weights of $p$ are equal) \cite{STECS,HD2017},
 \emph{free-choice} nets 
 (the net is homogeneous, and any two transitions sharing an input place have the same set of input places) \cite{DesEsp,STECS},
choice-free nets (each place has at most one output transition) \cite{tcs97,PN14}, 
marked graphs (each place has at most one output transition and one input transition) \cite{chep,WTS92,TCS17,DH2018},
join-free nets (each transition has at most one input place) \cite{STECS,ACSD13,TECS14,HD2017}, etc.

More recently, 
another kind of restriction has been considered, 
limiting the number of different transition labels of the LTS in combination with restrictions on the LTS structure: 
for the binary case, feasibility of net synthesis from finite linear LTS and LTS with 
\emph{cycles}\footnote{A set $A$ of $k$ arcs in a LTS $G$ defines a cycle of $G$ 
if the elements of $A$ can be ordered as a sequence $a_1 \ldots a_k$
such that,
for each $i \in \{1, \ldots, k\}$, 
$a_i = (n_i,\ell_i,n_{i+1})$
and
$n_{k+1} = n_1$, 
i.e.
the $i$-th arc $a_i$
goes from node $n_i$ to node $n_{i+1}$ until the first node $n_1$ is reached, closing the path. 
Cycles are also sometimes called circuits, circles and oriented cycles.} 
has been characterised by rates of labels 
in the transition system \cite{BarylskaBEMP15,BarylskaBEMP16} 
and 
by pseudo-regular expressions \cite{ErofeevBMP16}, 
giving rise to fast specialised synthesis algorithms;
moreover,
a complete enumeration of the shapes 
of synthesisable transition systems is presented in \cite{ErofeevW17}.

In this paper,
we combine the restriction on the number of labels
with the weighted marked graph (WMG) constraint.
In addition,
we study constraints on the existence
of cycles in the LTS: 
when the LTS is \emph{acyclic},
i.e. it does not contain any 
\emph{cycle},
and when it is \emph{cyclic},
i.e. it contains at least one cycle.
In the latter case,
we also study the finite \emph{circular LTS}, 
meaning strongly connected LTSs that contain a unique cycle: 
we investigate the \emph{cyclic solvability} of a word $w$, 
meaning the existence of a Petri net solution
to the finite circular LTS induced by the infinite \emph{cyclic word} $w^\infty$.

An important purpose of studying
such constrained LTSs
is to better understand the relationship 
between LTS decompositions and their solvability by Petri nets.
Indeed,
the unsolvability of simple subgraphs of the given LTS,
typically elementary paths (i.e. not containing any node twice)
and cycles (i.e. closed paths, whose start and end states are equal), 
often induces simple conditions of unsolvability for the entire LTS, 
as highlighted in other works \cite{BarylskaBEMP15,ErofeevBMP16,besdev-ccc25}.
Moreover, 
cycles appear systematically in the reachability graph  
of live and/or reversible Petri nets 
\cite{tcs97,Hujsa2016},
which are used to model various real-world applications,
such as embedded systems \cite{phdhujsa}.\\

{\bf Contributions.}
In this work, 
we study further the links between simple LTS structures and the reachability graph of WMGs,
as follows.\\
First,
we provide a characterisation of the $2$-label 
(i.e. binary) words being 
cyclically solvable by a WMG 
(i.e. WMG-solvable), 
and extend the analysis to finite cyclic LTSs. 
We also tackle the case of infinite cyclic LTSs with $2$ labels.

Then, 
when the number of labels is arbitrary,
we provide 
a geometric characterisation
of the finite, acyclic, WMG-solvable LTS, 
as well as
a general sufficient condition 
of WMG-solvability for a cyclic word,
using a decomposition into specific 
cyclically WMG-solvable binary subwords. 
We prove that this sufficient condition 
becomes a characterisation of cyclic 
WMG-solvability 
for a subclass of the $3$-label words.
Furthermore,
we show,
with the help of two counter-examples,
that this characterisation
does not hold
for words with four or five labels. 

Comparing with \cite{AtaedDEH18},
we refine the results and explanations on 
WMG-solvable, finite, cyclic, binary LTSs
by introducing
Lemma~\ref{Binary-states.lem}
and upgrading Theorem~\ref{cyc.thm},
in Subsection~\ref{circularbinary.subsec}.
We also provide 
the new geometric characterisation
of WMG-solvability for acyclic LTS
with any number of labels, 
and
we sharpen the counter-examples to the 
characterisation of cyclically solvable ternary words in the cases of four and five labels.\\

{\bf Organisation of the paper.}
After recalling classical definitions, notations and properties in Section \ref{Def.sec},
we present the results of WMG-solvability 
for $2$-label words in Section \ref{twolabels.sec}.
Then,
in Section \ref{klabels-acyc.sec},
we propose 
the geometric characterisation
of WMG-solvability for an acyclic LTS
with any number of labels. 
In Section \ref{klabels-cyc.sec},
we develop the general sufficient condition of circular WMG-solvability 
for any number of labels.
In Section \ref{threelabels.sec},
we tackle the ternary case 
and exhibit counter-examples for 
$4$ and $5$ labels. 
Finally,
Section \ref{conclu.sec} 
presents our conclusions and perspectives.

\section{Classical Definitions, Notations and Properties}\label{Def.sec}

{\bf LTSs, sequences and reachability.} 
A {\em labelled transition system with initial state}, 
abbreviated {\em LTS}, 
is a quadruple $\TS=(S,\to,T,\is)$ 
where $S$ is the set of {\em states}, 
$T$ is the set of {\em labels}, 
$\to\,\subseteq(S\times T\times S)$ is the {\em labelled transition relation}, 
and 
$\is\in S$ is the {\em initial state}.\\
A label $t$ is {\em enabled} at $s\in S$, written $s[t\rangle$, if $\exists s'\in S\colon(s,t,s')\in\to$, 
in which case $s'$ is said to be {\em reachable} from $s$ by the firing of $t$, and we write $s[t\rangle s'$.
Generalising to any (firing) sequences $\sigma\in T^*$,
$s[\emptyseq\rangle$ and $s[\emptyseq\rangle s$ are always true;
and $s[\sigma t\rangle s'$, i.e.  $\sigma t$ is {\em enabled} from state $s$ and leads to $s'$,
if there is some $s''$ 
with $s[\sigma\rangle s''$ and $s''[t\rangle s'$.
A state $s'$ is {\em reachable} from state $s$ if $\exists\sigma\in T^*\colon s[\sigma\rangle s'$.
The set of states reachable from $s$ is denoted by $[s\rangle$.\\
 
{\bf Petri nets, reachability and languages	.} 
A (finite, place-transition) 
\emph{weighted Petri net}, or \emph{weighted net},
is a tuple $N=(P,T,W)$ where
$P$ is a finite set of {\em places}, 
$T$ is a finite set of {\em transitions}, with $P\cap T=\es$,
and $W\colon((P\times T)\cup(T\times P))\to\nsymbol$ is a {\em weight} function
giving the weight of each arc.
A \emph{Petri net system}, or \emph{system},
is a tuple $\mathcal S=(N,M_0)$ where 
$N$ is a net
and
$M_0$ is the {\em initial marking},
which is a mapping $M_0\colon P\to\nsymbol$ (hence a member of $\nsymbol^P$)
indicating the initial number of {\em tokens} in each place. 
If $W(x,y)>0$, $y$ is said to be an output of $x$, and $x$ is said to be an input of $y$. 
The {\em incidence matrix} $C$ of the net is the integer $P\times T$-matrix with components 
$C(p,t)=W(t,p)-W(p,t)$.

A transition $t\in T$ is {\em enabled by} a marking $M$, 
denoted by $M[t\rangle$, if for all places $p\in P$, $M(p)\geq W(p,t)$.
A place $p \in P$ is {\em enabled by} a marking $M$ 
if $M(p) \ge W(p,t)$ for every output transition $t$ of $p$,
meaning that it is not an obstacle to enabling transitions. 
If $t$ is enabled at $M$, then $t$ can {\em occur} (or {\em fire}) in $M$, 
leading to the marking $M'$ defined by $M'(p)=M(p)-W(p,t)+W(t,p)$;
we note $M[t\rangle M'$.
A marking $M'$ is {\em reachable} from $M$ if there is a sequence of firings leading
from $M$ to $M'$; if this sequence of firings  defines a sequence of transitions $\sigma\in T^*$,
we note $M[\sigma\rangle M'$. 
The set of markings reachable from $M$ is denoted by $[M\rangle$.
The {\em reachability graph of $\mathcal S$} is the labelled transition system $\mathit{RG}(\mathcal S)$ 
with the set of vertices $[M_0\rangle$, the set of labels $T$, initial state $M_0$ 
and transitions $\{(M,t,M')\mid M,M'\in[M_0\rangle\land M[t\rangle M'\}$.
A system $\mathcal S$ is {\em bounded} if $\mathit{RG}(\mathcal S)$ is finite.

The language of a Petri net system $\mathcal S$ 
is the set $\lge(\mathcal S)=\{\sigma\in T^*\mid M_0[\sigma\rangle\}$. 
These languages are prefix-closed, 
i.e., if $\sigma=\sigma'\sigma''\in\lge(\mathcal S)$, then $\sigma'\in\lge(\mathcal S)$. 
For any language $L\subseteq T^*$, 
we denote by $\pref(L)$ the language formed by its prefixes.\\

{\bf Vectors.}
The {\em support}
of a vector is the set of the indices of its non-null components.
Consider any net $N=(P,T,W)$ with its incidence matrix~$C$.
A {\em T-vector} is an element of $\nsymbol^{T}$; 
it is called {\em prime} if 
the greatest common divisor of its components is one 
(i.e. its components do not have a common non-unit factor). 
A {\em T-semiflow} $\nu$ of the net is a non-null T-vector 
such that $C\cdot\nu=\zero$. 
A T-semiflow is called {\em minimal} when it is prime and its support 
is not a proper superset of the support of any other T-semiflow \cite{tcs97}.\\ 
The {\em Parikh vector} $\Parikh(\sigma)$ 
of a finite sequence $\sigma$ of transitions is a T-vector
counting the number of occurrences of each  transition in $\sigma$,
and the {\em support} of $\sigma$ is the support of its Parikh vector, 
i.e.
$\support(\sigma)=\support(\Parikh(\sigma))=\{t\in T\mid\Parikh(\sigma)(t)>0\}$.\\

{\bf Strong connectedness and cycles in an LTS.}
The LTS $(S,\to,T,\is)$ is said {\em reversible} if,
$\forall s\in[\is\rangle$, 
we have $\is\in[s\rangle$, i.e., 
it is always possible to go back to the initial state;
reversibility implies the strong connectedness of the LTS.\\
A sequence $s[\sigma\rangle s'$ is called a {\em cycle}, or more precisely 
a {\em cycle at (or around) state $s$}, if $s=s'$.
A non-empty cycle $s[\sigma\rangle s$
is called {\em small} 
if there is no non-empty cycle 
$s'[\sigma'\rangle s'$ in $\TS$ with
$\Parikh(\sigma') \lneqq \Parikh(\sigma)$,
meaning that no component of the left vector
is greater than the corresponding component of the right vector,
and
at least one is smaller
(the definition of Parikh vectors 
extending readily to sequences 
over the set of labels $T$ of the LTS).\\
A \emph{circular LTS} is a finite, strongly connected LTS 
that contains a unique cycle; hence, it has the shape of an oriented circle.
The circular LTS \emph{induced by}
a word $w\!=\!w_1\ldots w_k$ is the LTS 
with initial state $s_0$
defined as
$s_0[w_1\rangle s_1 [w_2\rangle s_2 \ldots [w_k\rangle s_0$.

All notions defined for labelled transition systems apply to
Petri nets through their reachability graphs.\\

{\bf Petri net subclasses.} 
A net $N$ is {\em plain} 
if no arc weight exceeds $1$;  
{\em pure} 
if $\forall p\in P\colon(p^\dt{\cap}{}^\dt p)=\es$, 
where $p^\dt=\{t\in T\mid W(p,t){>}0\}$ and ${}^\dt p=\{t\in T\mid W(t,p){>}0\}$;
CF ({\it choice-free} \cite{crespi-mandrioli-75,tcs97}) 
or ON (place-output-nonbranching \cite{besdev-ccc25})
if $\forall p\in P\colon|p^\dt|\leq1$; 
a  WMG ({\em weighted marked graph} \cite{WTS92}) 
if $|p^\dt|\leq 1$ and $|{}^\dt p|\leq 1$ for all places $p\in P$.
The latter form a subclass of the choice-free nets;
other subclasses are {\em marked graphs} \cite{chep},
which are plain with $|p^\dt|=1$ and $|{}^\dt p|=1$ for each place $p\in P$, 
and {\em T-systems} \cite{DesEsp}, which are plain  
with 
$|p^\dt|\leq 1$ and $|{}^\dt p|\leq 1$ 
for each place $p\in P$.\\

{\bf Isomorphism and solvability.} 
Two LTS $\TS_1=(S_1,\to_1,T,s_{01})$ and $\TS_2=(S_2,\to_2,T,s_{02})$ are 
isomorphic if there is a bijection $\zeta\colon S_1\to S_2$ with 
$\zeta(s_{01})=s_{02}$ and 
$(s,t,s')\in\to_1\,\Leftrightarrow(\zeta(s),t,\zeta(s'))\in\to_2$, for all $s,s'\in S_1$.\\
If an LTS $\TS$ is isomorphic to $\mathit{RG}(\mathcal S)$ where $\mathcal S$ is a system,
we say that $\mathcal S$ {\em solves} $\TS$.
Solving a word $w=\ell_1 \ldots \ell_k$ amounts to
solve the acyclic LTS defined by the single path $\is [\ell_1\rangle s_1 \ldots [\ell_k\rangle s_k$. 
A finite word $w$ is \emph{cyclically solvable}
if the circular LTS induced by $w$ is solvable.
A LTS is WMG-solvable if a WMG solves it.\\

{\bf Other classical notions.}  An LTS $\TS = (S,\rightarrow, T,\is)$ is \emph{fully reachable} if $S=[\is\rangle$.
It is  \emph{forward deterministic} if 
$s[t\rangle s'\land s[t\rangle s''\impl s'=s''$, 
and
 \emph{backward deterministic} if $s'[t\rangle s\land s''[t\rangle s\impl s'=s''$.\\
A system $\mathcal S$ is  \emph{forward persistent} if,
for any reachable markings $M,M_1,M_2$, 
$(M[a\rangle M_1 \land M[b\rangle M_2 \land a\neq b) \impl M_1[b\rangle M' \land M_2[a\rangle M'$
 for a reachable marking $M'$;
 it is  \emph{backward persistent} if,
for any reachable markings $M,M_1,M_2$, 
$(M_1[a\rangle M \land M_2[b\rangle M \land a\neq b) \impl M'[b\rangle M_1 \land M'[a\rangle M_2$ 
for a reachable marking $M'$.\\

Next, we recall classical properties of Petri net reachability graphs.

\begin{proposition}[Classical Petri net properties]\label{classic.prop} 
If $\mathcal S$ is a Petri net system: \\
$-$ $\mathit{RG}(\mathcal S)$ is a fully reachable LTS.\\
$-$ $\mathit{RG}(\mathcal S)$ is forward deterministic
and
backward deterministic.
\end{proposition}

For the subclass of WMGs, we have the following dedicated properties,
extracted from Proposition $4$, Lemma $1$, Theorem $2$ and Lemma $2$ in \cite{DH2018}.

\begin{proposition}[Properties of WMG] \label{WMG.prop} 
If $\mathcal S=(N,M_0)$ is a WMG system:\\ 
$-$ It is forward  persistent and backward persistent.\\
$-$ If $N$ is connected and has a T-semiflow $\nu$,	
then there is a unique minimal one $\pi$, with support $T$,
and $\nu = k \cdot \pi$ for some positive integer $k$. 
Moreover, 
if there is a non-empty cycle in $\mathit{RG}(\mathcal S)$, 
there is one with Parikh vector $\pi$ in $\mathit{RG}(\mathcal S)$ around each reachable marking 
and  $\mathit{RG}(\mathcal S)$ is reversible.
If there is no cycle, 
all the paths starting from some state $s$ and reaching some state $s'$ have the same Parikh vector. 
\end{proposition}

To simplify our reasoning in the sequel,  
we introduce the following notation,
which captures some of the behavioural properties satisfied by WMG
(Propositions \ref{classic.prop} and \ref{WMG.prop}).
We denote by 
\vspace*{-3mm}
\begin{itemize}[label={--}]
\item \text{\bf b} (for basic) the set of properties: 
 forward and backward deterministic, 
forward and backward persistent, totally reachable;
\item \text{\bf c} (for cyclic) the property: there is a small cycle 
whose Parikh vector is prime with support $T$. 
\end{itemize}

A synthesis procedure does not necessarily lead to a connected solution.
However,
the technique of decomposition into prime factors 
described in \cite{dev-ACSD16,devacta17}
can always be applied first,
so as to handle connected partial solutions and recombine them afterwards. 
Hence,
in the following,
we focus on connected WMGs, without loss of generality.
In the next section,
 we consider the synthesis problem of WMG with exactly two different labels.

\section{Synthesis of a WMG from a Cyclic Binary LTS}\label{twolabels.sec} 

In this section, 
we provide conditions for the WMG-solvability of $2$-label cyclic LTS. 
In Subsection \ref{circularbinary.subsec},
we investigate the WMG-solvability of a finite cyclic LTS:
first when it is circular, then without this
constraint. 
In Subsection~\ref{infiniterev.subsec},
we investigate the WMG-solvability of an 
infinite cyclic binary LTS.

\subsection{WMG-solvable Finite Cyclic Binary LTS} \label{circularbinary.subsec}

In this subsection, 
we first consider any circular LTS with only two different labels.
Each such LTS is defined by a word $w\in\{a,b\}^*$, 
corresponding to the labels encountered by firing the circuit once from $\is$, leading back to $\is$.
Changing the initial state in this LTS amounts to rotate $w$.
Clearly,
each such LTS satisfies property {\bf b},
but is not always WMG- (nor even Petri net-) 
solvable.

The next results consider circuit Petri nets as represented in Fig. \ref{solcyc.fig}, 
where places are named following the direction of the arcs, 
e.g. $p_{a,b}$ is the output place of $a$
and the input place of $b$.

\begin{figure}[htbp]
\centering
\begin{tikzpicture}[scale=0.9]
\node[place,tokens=0](p)at(2,0.85)[label=below:$p_{a,b}$]{$i$};
\node[place,tokens=0](p')at(2,-0.85)[label=above:$p_{b,a}$]{$j$};
\node[transition](a)at(0,0){$a$};
\node[transition](b)at(4,0){$b$};
\draw(a)[]edge[-latex]node[above left]{$m$}(p);
\draw(p)[]edge[-latex]node[above right]{$n$}(b);
\draw(p')[]edge[-latex]node[below left]{$m$}(a);
\draw(b)[]edge[-latex]node[below right]{$n$}(p');
\end{tikzpicture}
\caption{A generic WMG solving a finite circular LTS induced by a word $w$ over the alphabet $\{a,b\}$,
whose initial marking $(i,j)$ depends on the given solvable LTS.
We assume that $\Parikh(w) = (n,m)$ is prime.
}
\label{solcyc.fig}
\end{figure}
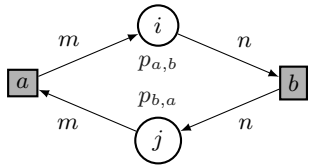

\begin{theorem}[Cyclically WMG-solvable binary words] 
\label{circularbinarywords.thm}
Consider a finite binary word $w$ over
the alphabet $\{a,b\}$,
with $\Parikh(w) = (n,m)$ 
and $n\leq m$,
the case $m\leq n$ being handled symmetrically.
Then,
$w$ is cyclically solvable 
if and only if 
$gcd(n,m) = 1$
and $w$ is a rotation of the word
$w' = ab^{m_0} \ldots ab^{m_{n-1}}$,
where the sequence $m_0, \ldots, m_{n-1}$
is the sequence of quotients
in the following
system of equalities, 
with $r_0 = 0$:
$$
\left\{
\begin{array}{l}
	r_0 +  m= m_0 \cdot  n + r_1, \textrm{~where } 0 \le r_1 <  n\\
	r_1 +  m = m_1 \cdot  n+ r_2, \textrm{~where } 0 \le r_2 <  n\\
	\ldots \\
	r_{n-1} +  m = m_{n-1} \cdot n.
\end{array}
\right.
$$
Moreover,
$m+n-1$ tokens
are necessary and sufficient to
solve the word cyclically.
\end{theorem}

\begin{proof}  
From Proposition \ref{WMG.prop},
for a connected WMG solution to exist,
the Parikh vector of the word must be the minimal T-semiflow $\mu=(n,m)$ with support $T=\{a,b\}$,
which is prime by definition,
thus $gcd(m,n)=1$.
A variant of this problem has been studied in \cite{binary16},
section 6.
Basing of this previous study,
we highlight the following facts,
leading to the claim.
If a solution exists, then:
\begin{itemize}[label={--}]
\item there exists a WMG solution 
as pictured in Fig. \ref{solcyc.fig},
in which
each firing preserves
the number of tokens;
thus,
denoting by $M_s(p)$ 
the marking of place $p$ at state $s$,
the sum $M_s(p_{a,b})+M_s(p_{b,a})$
is the same for all states.
\item Consider any two different  
reachable markings $M'$ and $M''$,
then,
from the above,
$M'(p_{a,b}) \neq M''(p_{a,b})$
and 
$M'(p_{b,a}) \neq M''(p_{b,a})$.
\item $M_s(p_{a,b})+M_s(p_{b,a})=m+n-1$.
Indeed, 
with more tokens, 
a reachable marking enables both $a$ and $b$, which is not allowed by the given LTS; 
with fewer tokens, 
a deadlock is reached, 
i.e. a marking that enables no transition.
\item For each $i$, $m_i\in\{\lfloor m/n \rfloor,\lceil m/n\rceil\}$,
there are $(m\!\!\!\mod\!n)$ $b$-blocks of size $(\lfloor m/n \rfloor\!+\!1)$, the other ones have size $\lfloor m/n \rfloor$. 
\end{itemize}

Let us start from the state $s$ such that $M_s(p_{a,b})=0$ and $M_s(p_{b,a})=m+n-1$, 
with
$r_0=0$. 
We denote by $r_i$ the number of tokens in $p_{a,b}$ at the $(i+1)$-th visited state that enables $a$.
The value
$m_0$ is the maximal number of $b$'s that can be fired after the first $a$, and then $r_1$ tokens remain in $p_{a,b}$;
hence, 
there are $m+n-1-r_1$ tokens in  $p_{b,a}$ (which is at least $m$) before the second $a$. 
After the second $a$, 
we have $m+r_1$ tokens in $p_{a,b}$ and we fire $m_1$ $b$'s.
We iterate the process until the initial state is reached. 

In the state enabling the $(i+1)$-th $a$,
there are $(i\cdot m)\mod n$ tokens in $p_{a,b}$,
implying that $r_n$ equals $0$ 
when the initial state is reached again.
In between, 
we visited all the values from $0$ to $n-1$ for the $r_i$'s:
indeed, 
if $(i\cdot m)\mod n=(j\cdot m)\mod n$ for $0\leq i<j<n$, we have $((j-i)\cdot m)\mod n=0$, or $((j-i)\cdot m=k\cdot n$ 
for some $k$; 
but then $n$ must divide $j-i$ since $m$ and $n$ are relatively prime, which is only possible if $i=j$.

Finally,
some rotation
of $w'$ leads to $w$ and to the associated value of $r_0$.
\qed \end{proof}

An example is given in Fig. \ref{binary-circuit.fig},
where the elements of the sequence $m_0, \ldots, m_{n-1}$ are put in bold in the system on the left.

\begin{figure}[htbp]
\centering
\begin{minipage}{0.4\linewidth}
$0+21={\bf 2}.8+5$\\
$5+21={\bf 3}.8+2$\\
$2+21={\bf 2}.8+7$\\
$7+21={\bf 3}.8+4$\\
$4+21={\bf 3}.8+1$\\
$1+21={\bf 2}.8+6$\\
$6+21={\bf 3}.8+3$\\
$3+21={\bf 3}.8+0$.\\
\end{minipage}
\raisebox{-1cm}{
\begin{tikzpicture}[scale=1]
\node[place,tokens=0](p1)at(2,1)[label=above:$p_1$]{$28$};
\node[place,tokens=0](p2)at(2,3)[label=below:$p_2$]{};
\node[transition](a)at(1,2){$a$};
\node[transition](b)at(3,2){$b$};
\draw(p1)[]edge[-latex,bend left=40]node[below left]{$21$}(a);
\draw(a)[]edge[-latex,bend left=40]node[above left]{$21$}(p2);
\draw(p2)[]edge[-latex,bend left=40]node[above right]{$8$}(b);
\draw(b)[]edge[-latex,bend left=40]node[below right]{$8$}(p1);
\end{tikzpicture}
}

\caption{This system solves the word $w = ab^{2}ab^{3}ab^{2}ab^{3}ab^{3}ab^{2}ab^{3}ab^{3}$ cyclically.
}
\label{binary-circuit.fig}
\end{figure}
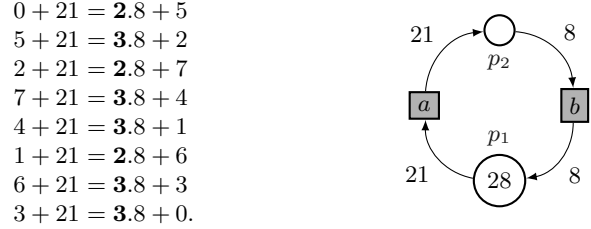

{\bf Complexity.} 
The number of operations 
to determine the sequence of $m_i$'s 
is linear in the smallest weight $n$, 
i.e. also in the minimal number of 
occurrences of a label.
In comparison, the previous algorithm of \cite{binary16}
checks a quadratic number of subwords.\\

The next lemma characterises
the set of states reachable 
in any WMG 
whose underlying net is the one pictured
in Fig. \ref{solcyc.fig}
and
whose initial marking
contains at least $n+m-1$ tokens.

\begin{lemma}[Reachable states w.r.t. 
the number of tokens]  
\label{Binary-states.lem}
Let $N$ be a binary WMG
as in Fig. \ref{solcyc.fig},
such that
$\mu=(n,m) \ge \one$ 
with $\gcd(n,m)=1$.
Then
for each positive integer $k \ge m+n-1$,
and each marking $M_0$ for $N$,
the following properties are equivalent:\\
$1)$
$M_0$
contains exactly $k$ tokens;\\
$2)$
RG$((N,M_0))$ contains exactly $k+1$ states;\\
$3)$
the set of states of RG$((N,M_0))$
is
$\{(k_{ab}, k_{ba}) \in \mathbb{N}^2 \,|\,  k_{ab} + k_{ba} = k\}$.
\end{lemma}

\begin{proof}
We first prove that $1)$ implies $2)$ and $3)$.

The case $k=m+n-1$
follows from the proof of
Theorem~\ref{circularbinarywords.thm}:
all the markings
of the form $(k_{ab},m+n-1-k_{ab})$
with $k_{ab} \in [0,m+n-1]$
are reachable.
By the preservation of the total number of
tokens through firings,
these markings are distinct
and their amount is thus $n+m$,
proving
the claim for $k=m+n-1$.

In the following,
let us denote by $S_\bot$
the set of all these markings (i.e. with exactly $m+n-1$ tokens).

The case $k = \ell\cdot (m+n-1)$ 
for some positive integer $\ell$ 
is deduced similarly:
denoting $M_0$ as any sum
of $\ell$ markings $M_1 \ldots M_\ell$
such that each $M_i$ corresponds to some distribution
of $k$ tokens over the two places,
firing sequences are allowed in each $(N,M_i)$
that lead to all the markings of 
$S_\bot$,
i.e. $S_\bot$ is the set of states of RG($(N,M_i)$) for each $i$
and all these RG's differ only by the choice of the initial state.
All these sequences,
obtained from all $i \in [1,\ell]$,
are allowed independently (sequentially as well as in a shuffle) 
 in $(N,M_0)$.
Thus,
all the markings of the form 
$(k_{ab},k-k_{ab})$
with $k_{ab} \in [0,k]$
are mutually reachable,
describing $k+1$ distinct markings,
which correspond to all the possibilities
of distributing $k$ tokens over
the two places.

Now,
let us consider $k>m+n-1$,
denoting the initial marking as
$M_0 = (u+u',v+v')$ such that
$u+v = \ell\cdot (m+n-1)$, 
$\ell \in \mathbb{N}_{>0}$,
and $u',v'$ are non-negative integers with
$m+n-1 > u'+v' \ge 1$.
From the above,
all the markings of the form $M+(u',v')$,
where $M$ belongs to $RG(N,(u,v))$
are reachable from $M_0$,
describing $u+v+1$ distinct markings.
Other markings can be reached by firing
the tokens of $u'$ and $v'$:
for each $x \in [0;u'-1]$
and each $y \in [0;v'-1]$,
the markings $(x+n,k-x-n)$ and $(k-y-m,y+m)$ are reachable,
from which we may fire $b$ and $a$, respectively, leading to markings
$(x,k-x)$ and $(k-y,y)$, 
and all these markings are distinct.
Thus, 
we reach at least $u+v+1+u'+v'=k+1$
distinct markings,
which describe all the 
possible distributions of $k$.

We deduce that $1)$ implies $2)$ and $3)$.
Now,
assuming $2)$, and from the reasoning above,
$M_0$ cannot have strictly less nor strictly more
than $k$ tokens, implying $1)$.
Finally, 
$3)$ describes all the $k+1$ distributions of the $k$ tokens
over the two places,
all of these markings being mutually reachable by
reversibility,
hence $1)$ and $2)$ are obtained.
\qed
\end{proof}

In Theorem \ref{circularbinarywords.thm}, 
we provided a criterion for the cyclic solvability of a given word.
In the next theorem,
we abstract the word by a Parikh vector,
which provides less accurate information on the behaviour of the process.
This result investigates the possible 
WMG-solvable LTS
for this vector.

\begin{theorem}[WMG-solvable reversible binary LTS] 
\label{cyc.thm}
Let us consider
$\mu=(n,m) \ge \one$ 
such that $\gcd(n,m)=1$,
and a positive integer $k$.  
Up to isomorphism and the choice of the initial state, 
when $k \geq n+m$,  
there exists a single finite WMG-solvable LTS 
$(S,\to,\{a,b\},\is)$
that satisfies \text{\bf b}, 
\text{\bf c} and $|S| = k$,  
and that contains a small cycle whose Parikh vector is $\mu$.
No such WMG-solvable LTS exists when 
$k < n+m$.  
In the particular case of
$S=\{0,1,\ldots,m+n-1\}$, we have (up to isomorphism) $\to=\{(i,a,i+m)|i,i+m\in S\}\cup\{(i,b,i-n)|i,i-n\in S\}$.
\end{theorem}
\begin{proof}  
If a solution exists, 
it has the form of Fig. \ref{solcyc.fig}.
If $k \ge n+m$,
there are exactly $k-1$ tokens in the system
by Lemma \ref{Binary-states.lem}
and the reachability graph is unique 
up to isomorphism.
From the previous results of this section,
if $M_0=n+m-1$,
then
the RG is circular and contains
exactly $n+m$ distinct states:
all the values for $i$ between $0$ and $n+m-1$ are reached in some order.
Moreover,
if we identify the states to $i$,
i.e., 
the marking of $p_{a,b}$, 
the arcs are 
$\{(i,a,i+m)|0\leq i,i+m<n+m\in S\}\cup\{(i,b,i-n)|0\leq i,i-n<n+m \in S\}$.
As a consequence, 
if $|S|<n+m$, there aren't enough states to close the circuit, 
and there is no solution.
The rest of the claim immediately results from Lemma~\ref{Binary-states.lem}.
\qed \end{proof}

\subsection{WMG-solvable Infinite Cyclic Binary LTS}\label{infiniterev.subsec}

Let us consider an infinite LTS satisfying \text{\bf b} and \text{\bf c} with only two different labels.
From the previous section,
it cannot correspond to a net of the kind illustrated in Fig.~\ref{solcyc.fig}
since $i+j$ remains constant, hence yields finitely many states.
On the other hand, 
a net of the kind illustrated
in Fig.~\ref{isolcyc.fig},
or the variant obtained by switching
the roles of $a$ and $b$,
yields infinitely many occurrences
of transition $a$, 
leading to infinitely
many different reachable markings.
Besides,
from any state, there may only be finitely many consecutive $b$'s.
Moreover, 
this is the only way to obtain infinitely many cycles with Parikh vector $(n,m)$. 

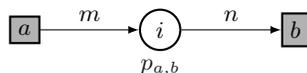
\begin{figure}[htbp]
\begin{center}
\begin{tikzpicture}[scale=0.9]
\node[place,tokens=0](p)at(2,0)[label=below:$p_{a,b}$]{$i$};
\node[transition](a)at(0,0){$a$};
\node[transition](b)at(4,0){$b$};
\draw(a)[]edge[-latex]node[above]{$m$}(p);
\draw(p)[]edge[-latex]node[above]{$n$}(b);
\end{tikzpicture}
\end{center}
\caption{A WMG solution for the infinite cyclic case.}
\label{isolcyc.fig}
\end{figure}

If $n=1$, $i$ is the maximum number of consecutive executions of $b$ from $\is$; we can then verify 
if the given LTS corresponds to the constructed net.
Otherwise, let $k$ and $l$ be the Bezout coefficients corresponding to the relatively prime numbers $m$ and $n$, 
so that $k\cdot m+l\cdot n=1$. 
If $l\geq 0\geq k$, $i$ is the maximum number of times we may execute $a^{-k}b^{l}$ 
consecutively from $\is$, 
and 
we can check again if the given LTS corresponds to the constructed net (this is a direct generalisation of the case $n=1$).
Otherwise, 
since $-n\cdot m+m\cdot n=0$, by adding this relation enough times to the previous one, we get 
$k'\cdot m+l'\cdot n=1$ with $l'\geq 0\geq k'$, and we apply the same idea.

\section{WMG-solvable Acyclic LTS: a Geometric Approach}\label{klabels-acyc.sec}

In what follows,
we consider any 
acyclic LTS satisfying property \text{\bf b}.
First,
in Subsection \ref{acyc2.sct},
we give a geometric interpretation of 
WMG-solvability
for acyclic LTS with only two different labels.
Then,
in Subsection \ref{acyck.sct},
we extend this result to any number of labels.

\subsection{Geometric Characterisation for 2 Labels}\label{acyc2.sct}

In the following,
we specialise to the WMG case   
the more general framework considered
in \cite{ErofeevW17}, Theorem 2,
using convex sets of $\nsymbol^2$.
The standard definition of convex sets of 
$\rsymbol^2$ is given by the 
segment-inclusion property:
a set $C\subseteq \rsymbol^2$ is convex
if and only if, 
for any $x,y\in C$, $[x,y]\subseteq C$, 
where $[x,y]$ is the linear segment with extremities $x$ and $y$.
However, 
this does not work for $\nsymbol^2$ (nor $\zsymbol^2$), 
as illustrated by Fig. \ref{conv.fig}:
in the set $C=\{x,y,z\}$ with $x=(0,0)$, $y=(1,2)$ and $z=(2,1)$, 
we have $[x,y]=\{x,y\}$, $[y,z]=\{y,z\}$ and $[z,x]=\{z,x\}$;
hence we have the segment-inclusion property;
however, 
clearly, 
this set should not be considered
as convex since 
the node $X=(1,1)$ is missing.

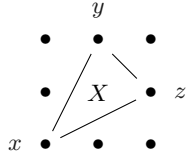
\begin{figure}[!ht]
\centering
\begin{tikzpicture}[scale=0.7]
\node[](x)at(0,0)[label=left:$x$]{$\dt$};\node[]at(1,0){$\dt$};\node[]at(2,0){$\dt$};
\node[]at(0,1){$\dt$};\node[]at(1,1){$X$};\node[](z)at(2,1)[label=right:$z$]{$\dt$};
\node[]at(0,2){$\dt$};\node[](y)at(1,2)[label=above:$y$]{$\dt$};\node[]at(2,2){$\dt$};
\draw[-](x)edge(y);\draw[-](y)edge(z);\draw[-](z)edge(x);
\end{tikzpicture}
\caption{Non-convex set in $\zsymbol^2$ 
with the segment-inclusion property.}
\label{conv.fig}
\end{figure}

In \cite{doignon73},
two equivalent definitions of convex sets 
in lattices like $\zsymbol^2$
are provided, 
which immediately extend to $\nsymbol^2$:
\begin{enumerate}
\item either as the intersection of a convex set of $\rsymbol^2$ with $\zsymbol^2$,
\item or as the intersection of half planes ${\mathcal L}_i$, with 
 ${\mathcal L}_i=\{(x,y)\in\zsymbol^2|a_i\cdot x+b_i\cdot y\geq c_i\mbox{ for some $a_i,b_i,c_i\in\zsymbol$}\}$.
If the convex set is finite,
we can use a finite set of such half-planes,
otherwise we may need (countably) infinitely many of them (notice that infinite convex sets exist
with a boundary defined by finitely many half-planes).
\end{enumerate}

In order to characterise the acyclic LTS  
that are solvable by WMG nets with two labels,
we first
identify isomorphically
each state $s$ with: \\
$\Delta_s=(\mbox{\small number of $a$'s in any path from $\is$ to $s$},
\mbox{\small number of $b$'s in any path from $\is$ to $s$})$ 
(this is coherent from Proposition~\ref{WMG.prop}), 
which amounts to consider for $S$ a part of $\nsymbol^2$
containing $(0,0)$ (=$\is$). 
From the full reachability, 
$S$ is connected 
(there is a directed path from $(0,0)$ 
to any $(i,j)\in S$, 
hence an undirected path between any two states). 
From Keller's theorem \cite{keller} (due to determinism and persistence), 
full reachability and Proposition~\ref{WMG.prop},
we have that 
$(i,j)\stackrel{a}{\rightarrow}(i',j')\iff i'=i+1\land j'=j$ 
and 
$(i,j)\stackrel{b}{\rightarrow}(i',j')\iff i'=i\land j'=j+1$. 

\begin{figure}[htbp]
\centering
\begin{tikzpicture}[scale=0.9]
\node[place,tokens=0](p)at(2,0)[label=below:$p_{a,b}$]{$m_{a,b}$};  
\node[transition](a)at(0,0){$a$};
\node[transition](b)at(4,0){$b$};
\draw(a)[]edge[-latex]node[above]{$W_a$}(p);
\draw(p)[]edge[-latex]node[above]{$W_b$}(b);
\node[place,tokens=0](p')at(8,0)[label=below:$p_{a,*}$]{$m_{a,*} $};  
\node[transition](a')at(6,0){$a$};
\draw(a')[]edge[-latex]node[above]{$W_a$}(p');
\node[place,tokens=0](p'')at(10,0)[label=below:$p_{*,b}$]{$m_{*,b}$}; 
\node[transition](b'')at(12,0){$b$};
\draw(p'')[]edge[-latex]node[above]{$W_b$}(b'');
\end{tikzpicture}
\caption{General places for a WMG synthesis, 
with initial marking $m_{a,b} = M_0(p_{a,b})$, 
$m_{a,*} = M_0(p_{a,*})$ 
and
$m_{*,b} = M_0(p_{*,b})$.
}
\label{place.fig}
\end{figure}
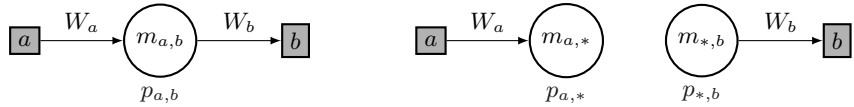

If the system is WMG-solvable, 
it must be defined by a finite set 
of places of the kind $p_{a,b}$ and 
$p_{b,a}$ in Fig. \ref{place.fig}
(with $\gcd(W_a,W_b)=1$ 
and 
$m_{a,b}, m_{b,a}\geq 0$), 
including the special cases $p_{*,b}$
(with $W_a=0$, $W_b=1$ and $m_{*,b}>0$, 
the case $m_{*,b}=0$ only serving to make $b$ non-firable 
but we assumed the system weakly live) 
or $p_{*,a}$, and $p_{a,*}$ 
(with $W_b=0$, $W_a=1$ and $m_{*,a}=m_{a,*}=0$) 
or $p_{b,*}$.
For a place $p_{a,b}$, 
we have for each state $s=(i,j)$ 
that the corresponding marking is 
$M_s(p_{a,b})=M_0(p_{a,b})+i\cdot W_a-j\cdot W_b$,
and since we must have $M_s(p_{a,b})\geq 0$, this defines a `region', both in the sense of \cite{bbd} and in an intuitive geometric meaning:
$R_{p_{a,b}}=\{(i,j)|M_0(p_{a,b})+W_a\cdot i-W_b\cdot j\geq 0\}$ or,
permuting the roles of $a$ and $b$,
$R_{p_{b,a}}=\{(i,j)|M_0(p_{b,a})-W_a\cdot i+W_b\cdot j\geq 0\}$,  
i.e. in either case the intersection of $\nsymbol^2$ with a half plane of $\zsymbol^2$.
These regions will be called in the following WMG-regions.
Notice that 
\cite{ErofeevW17} considers additional regions,  
where $W_a<0$ or $W_b<0$.
Each such region is convex, 
as well as any intersection of such regions. 

We deduce the next specialisation  
of Theorem 2 in \cite{ErofeevW17}.

\begin{theorem}[WMG-solvable acyclic binary LTS]
\label{acyc2.thm}
An acyclic  LTS satisfying property \text{\bf b} is WMG-solvable if and only if, when applied on 
$\nsymbol^2$, its set of states $S$ is 
connected, convex and delimited by 
(i.e., it is the intersection of) 
a finite set of WMG-regions.
A possible solution is then provided
by the places corresponding to these regions.
\end{theorem}

For any finite LTS, 
if it is the intersection of WMG-regions, 
it is the intersection of a finite set 
of such regions.
However, 
the result may be extended to an infinite LTS, 
but then it may be necessary to specify
that only a finite set of regions is allowed.
This is illustrated
by Fig.~\ref{synth-noncirc.fig}.

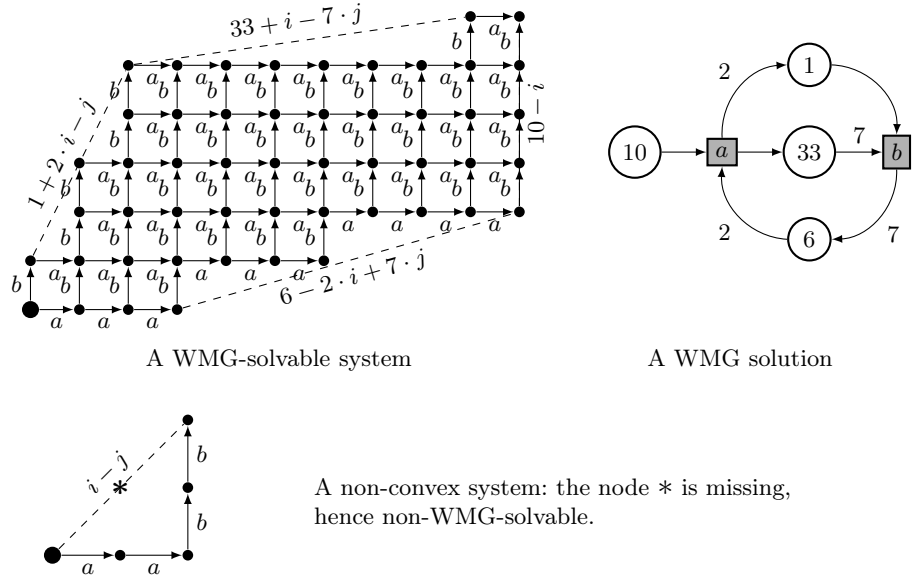
\begin{figure}[!ht]
\centering
\begin{tabular}{cc}
\begin{tikzpicture}[scale=0.65]
\node[circle,fill=black!100,inner sep=0.08cm](00)at(0,0)[]{};
\node[circle,fill=black!100,inner sep=0.05cm](10)at(1,0)[]{};
\node[circle,fill=black!100,inner sep=0.05cm](20)at(2,0)[]{};
\node[circle,fill=black!100,inner sep=0.05cm](30)at(3,0)[]{};
\node[circle,fill=black!100,inner sep=0.05cm](01)at(0,1){};
\node[circle,fill=black!100,inner sep=0.05cm](11)at(1,1){};
\node[circle,fill=black!100,inner sep=0.05cm](21)at(2,1){};
\node[circle,fill=black!100,inner sep=0.05cm](31)at(3,1){};
\node[circle,fill=black!100,inner sep=0.05cm](41)at(4,1){};
\node[circle,fill=black!100,inner sep=0.05cm](51)at(5,1){};
\node[circle,fill=black!100,inner sep=0.05cm](61)at(6,1){};
\node[circle,fill=black!100,inner sep=0.05cm](12)at(1,2)[]{};
\node[circle,fill=black!100,inner sep=0.05cm](22)at(2,2)[]{};
\node[circle,fill=black!100,inner sep=0.05cm](32)at(3,2)[]{};
\node[circle,fill=black!100,inner sep=0.05cm](42)at(4,2)[]{};
\node[circle,fill=black!100,inner sep=0.05cm](52)at(5,2)[]{};
\node[circle,fill=black!100,inner sep=0.05cm](62)at(6,2)[]{};
\node[circle,fill=black!100,inner sep=0.05cm](72)at(7,2)[]{};
\node[circle,fill=black!100,inner sep=0.05cm](82)at(8,2)[]{};
\node[circle,fill=black!100,inner sep=0.05cm](92)at(9,2)[]{};
\node[circle,fill=black!100,inner sep=0.05cm](102)at(10,2)[]{};
\draw[-latex](00)to node[auto,below]{$a$}(10);
\draw[-latex](10)to node[auto,below]{$a$}(20);
\draw[-latex](20)to node[auto,below]{$a$}(30);
\draw[-latex](00)to node[auto,left]{$b$}(01);
\draw[-latex](10)to node[auto,left]{$b$}(11);
\draw[-latex](20)to node[auto,left]{$b$}(21);
\draw[-latex](30)to node[auto,left]{$b$}(31);
\draw[-latex](01)to node[auto,below]{$a$}(11);
\draw[-latex](11)to node[auto,below]{$a$}(21);
\draw[-latex](21)to node[auto,below]{$a$}(31);
\draw[-latex](31)to node[auto,below]{$a$}(41);
\draw[-latex](41)to node[auto,below]{$a$}(51);
\draw[-latex](51)to node[auto,below]{$a$}(61);
\draw[-latex](12)to node[auto,below]{$a$}(22);
\draw[-latex](22)to node[auto,below]{$a$}(32);
\draw[-latex](32)to node[auto,below]{$a$}(42);
\draw[-latex](42)to node[auto,below]{$a$}(52);
\draw[-latex](52)to node[auto,below]{$a$}(62);
\draw[-latex](62)to node[auto,below]{$a$}(72);
\draw[-latex](72)to node[auto,below]{$a$}(82);
\draw[-latex](82)to node[auto,below]{$a$}(92);
\draw[-latex](92)to node[auto,below]{$a$}(102);
\draw[-latex](11)to node[auto,left]{$b$}(12);
\draw[-latex](21)to node[auto,left]{$b$}(22);
\draw[-latex](31)to node[auto,left]{$b$}(32);
\draw[-latex](41)to node[auto,left]{$b$}(42);
\draw[-latex](51)to node[auto,left]{$b$}(52);
\draw[-latex](61)to node[auto,left]{$b$}(62);
\node[circle,fill=black!100,inner sep=0.05cm](13)at(1,3)[]{};
\node[circle,fill=black!100,inner sep=0.05cm](23)at(2,3)[]{};
\node[circle,fill=black!100,inner sep=0.05cm](33)at(3,3)[]{};
\node[circle,fill=black!100,inner sep=0.05cm](43)at(4,3)[]{};
\node[circle,fill=black!100,inner sep=0.05cm](53)at(5,3)[]{};
\node[circle,fill=black!100,inner sep=0.05cm](63)at(6,3)[]{};
\node[circle,fill=black!100,inner sep=0.05cm](73)at(7,3)[]{};
\node[circle,fill=black!100,inner sep=0.05cm](83)at(8,3)[]{};
\node[circle,fill=black!100,inner sep=0.05cm](93)at(9,3)[]{};
\node[circle,fill=black!100,inner sep=0.05cm](103)at(10,3)[]{};
\draw[-latex](13)to node[auto,below]{$a$}(23);
\draw[-latex](23)to node[auto,below]{$a$}(33);
\draw[-latex](33)to node[auto,below]{$a$}(43);
\draw[-latex](43)to node[auto,below]{$a$}(53);
\draw[-latex](53)to node[auto,below]{$a$}(63);
\draw[-latex](63)to node[auto,below]{$a$}(73);
\draw[-latex](73)to node[auto,below]{$a$}(83);
\draw[-latex](83)to node[auto,below]{$a$}(93);
\draw[-latex](93)to node[auto,below]{$a$}(103);
\draw[-latex](12)to node[auto,left]{$b$}(13);
\draw[-latex](22)to node[auto,left]{$b$}(23);
\draw[-latex](32)to node[auto,left]{$b$}(33);
\draw[-latex](42)to node[auto,left]{$b$}(43);
\draw[-latex](52)to node[auto,left]{$b$}(53);
\draw[-latex](62)to node[auto,left]{$b$}(63);
\draw[-latex](72)to node[auto,left]{$b$}(73);
\draw[-latex](82)to node[auto,left]{$b$}(83);
\draw[-latex](92)to node[auto,left]{$b$}(93);
\draw[-latex](102)to node[auto,left]{$b$}(103);
\node[circle,fill=black!100,inner sep=0.05cm](24)at(2,4)[]{};
\node[circle,fill=black!100,inner sep=0.05cm](34)at(3,4)[]{};
\node[circle,fill=black!100,inner sep=0.05cm](44)at(4,4)[]{};
\node[circle,fill=black!100,inner sep=0.05cm](54)at(5,4)[]{};
\node[circle,fill=black!100,inner sep=0.05cm](64)at(6,4)[]{};
\node[circle,fill=black!100,inner sep=0.05cm](74)at(7,4)[]{};
\node[circle,fill=black!100,inner sep=0.05cm](84)at(8,4)[]{};
\node[circle,fill=black!100,inner sep=0.05cm](94)at(9,4)[]{};
\node[circle,fill=black!100,inner sep=0.05cm](104)at(10,4)[]{};
\draw[-latex](24)to node[auto,below]{$a$}(34);
\draw[-latex](34)to node[auto,below]{$a$}(44);
\draw[-latex](44)to node[auto,below]{$a$}(54);
\draw[-latex](54)to node[auto,below]{$a$}(64);
\draw[-latex](64)to node[auto,below]{$a$}(74);
\draw[-latex](74)to node[auto,below]{$a$}(84);
\draw[-latex](84)to node[auto,below]{$a$}(94);
\draw[-latex](94)to node[auto,below]{$a$}(104);
\draw[-latex](23)to node[auto,left]{$b$}(24);
\draw[-latex](33)to node[auto,left]{$b$}(34);
\draw[-latex](43)to node[auto,left]{$b$}(44);
\draw[-latex](53)to node[auto,left]{$b$}(54);
\draw[-latex](63)to node[auto,left]{$b$}(64);
\draw[-latex](73)to node[auto,left]{$b$}(74);
\draw[-latex](83)to node[auto,left]{$b$}(84);
\draw[-latex](93)to node[auto,left]{$b$}(94);
\draw[-latex](103)to node[auto,left]{$b$}(104);
\node[circle,fill=black!100,inner sep=0.05cm](25)at(2,5)[]{};
\node[circle,fill=black!100,inner sep=0.05cm](35)at(3,5)[]{};
\node[circle,fill=black!100,inner sep=0.05cm](45)at(4,5)[]{};
\node[circle,fill=black!100,inner sep=0.05cm](55)at(5,5)[]{};
\node[circle,fill=black!100,inner sep=0.05cm](65)at(6,5)[]{};
\node[circle,fill=black!100,inner sep=0.05cm](75)at(7,5)[]{};
\node[circle,fill=black!100,inner sep=0.05cm](85)at(8,5)[]{};
\node[circle,fill=black!100,inner sep=0.05cm](95)at(9,5)[]{};
\node[circle,fill=black!100,inner sep=0.05cm](105)at(10,5)[]{};
\draw[-latex](25)to node[auto,below]{$a$}(35);
\draw[-latex](35)to node[auto,below]{$a$}(45);
\draw[-latex](45)to node[auto,below]{$a$}(55);
\draw[-latex](55)to node[auto,below]{$a$}(65);
\draw[-latex](65)to node[auto,below]{$a$}(75);
\draw[-latex](75)to node[auto,below]{$a$}(85);
\draw[-latex](85)to node[auto,below]{$a$}(95);
\draw[-latex](95)to node[auto,below]{$a$}(105);
\draw[-latex](24)to node[auto,left]{$b$}(25);
\draw[-latex](34)to node[auto,left]{$b$}(35);
\draw[-latex](44)to node[auto,left]{$b$}(45);
\draw[-latex](54)to node[auto,left]{$b$}(55);
\draw[-latex](64)to node[auto,left]{$b$}(65);
\draw[-latex](74)to node[auto,left]{$b$}(75);
\draw[-latex](84)to node[auto,left]{$b$}(85);
\draw[-latex](94)to node[auto,left]{$b$}(95);
\draw[-latex](104)to node[auto,left]{$b$}(105);
\node[circle,fill=black!100,inner sep=0.05cm](96)at(9,6)[]{};
\node[circle,fill=black!100,inner sep=0.05cm](106)at(10,6)[]{};
\draw[-latex](96)to node[auto,below]{$a$}(106);
\draw[-latex](95)to node[auto,left]{$b$}(96);
\draw[-latex](105)to node[auto,left]{$b$}(106);
\draw[dashed](30)to node[auto,below,sloped]{$6-2\cdot i+7\cdot j$}(102);
\draw[dashed](01)to node[auto,above,sloped]{$1+2\cdot i- j$}(25);
\draw[dashed](25)to node[auto,above,sloped]{$33+ i- 7\cdot j$}(96);
\draw[dashed](102)to node[auto,below,sloped]{$10- i$}(106);
\end{tikzpicture}
&
\hspace*{4mm}
\raisebox{10mm}{
\begin{tikzpicture}[scale=0.58]
\node[place,tokens=0](p)at(2,2)[]{$1$};
\node[place,tokens=0](p')at(2,-2)[]{$6$};
\node[place,tokens=0](p'')at(2,0)[]{$33$};
\node[place,tokens=0](pp)at(-2,0)[]{$10$};
\node[transition](a)at(0,0){$a$};
\node[transition](b)at(4,0){$b$};
\draw(a)[]edge[-latex,bend left=45]node[above left]{$2$}(p);
\draw(p)[]edge[-latex,bend left=45]node[above]{$$}(b);
\draw(p')[]edge[-latex,bend left=45]node[below left]{$2$}(a);
\draw(b)[]edge[-latex,bend left=45]node[below right]{$7$}(p');
\draw(a)[]edge[-latex]node[above]{$$}(p'');
\draw(p'')[]edge[-latex]node[above]{$7$}(b);
\draw(pp)[]edge[-latex]node[above]{$$}(a);
\end{tikzpicture}
}\\
A WMG-solvable system
&
A WMG solution
\\
\end{tabular}

\vspace*{5mm}
\begin{tabular}{cc}
\begin{tikzpicture}[scale=0.9]
\node[circle,fill=black!100,inner sep=0.08cm](00)at(0,0)[]{};
\node[circle,fill=black!100,inner sep=0.05cm](10)at(1,0)[]{};
\node[circle,fill=black!100,inner sep=0.05cm](20)at(2,0)[]{};
\node[circle,fill=black!100,inner sep=0.05cm](21)at(2,1)[]{};
\node[circle,fill=black!100,inner sep=0.05cm](22)at(2,2)[]{};
\draw[-latex](00)to node[auto,below]{$a$}(10);
\draw[-latex](10)to node[auto,below]{$a$}(20);
\draw[-latex](20)to node[auto,right]{$b$}(21);
\draw[-latex](21)to node[auto,right]{$b$}(22);
\draw[dashed](00)to node[auto,above,sloped]{$ i- j$}(22);
\node(qq)at(1,1)[]{\large $\boldsymbol{*}$};
\end{tikzpicture}
&
\hspace*{1cm}
\raisebox{10mm}{
\begin{minipage}{0.6\linewidth}
A non-convex system:
the node {\large $*$} is missing,\\
hence non-WMG-solvable.
\end{minipage}
}
\\
\end{tabular}
\caption{Illustration of Theorem \ref{acyc2.thm}.} 
\label{synth-noncirc.fig}
\end{figure}

\begin{figure}[!ht]
\centering
\begin{tikzpicture}[scale=1.7]
\node[circle,fill=black!100,inner sep=0.08cm](00)at(0,0)[label=left:$\is$]{};
\node[circle,fill=black!100,inner sep=0.05cm](10)at(1,0)[]{};
\node[circle,fill=black!100,inner sep=0.05cm](21)at(2,1)[]{};
\node[circle,fill=black!100,inner sep=0.05cm]()at(2,0)[]{};
\node[circle,fill=black!100,inner sep=0.05cm]()at(0,1)[]{};
\node[circle,fill=black!100,inner sep=0.05cm]()at(1,1)[]{};
\draw[-latex](00)to node[auto,below]{$a$}(10);
\draw[dashed](00)to node[auto,above,sloped]{$ i-2\cdot  j$}(21);
\draw[dashed](10)to node[auto,below,sloped]{$1- i+j$}(21);
\end{tikzpicture}\hspace{2cm}
\raisebox{2mm}{
\begin{tikzpicture}[scale=0.75]
\node[place,tokens=0,label=below:$p_2$](p)at(2,1)[]{}; 
\node[place,tokens=1,label=above:$p_1$](p')at(2,-1)[]{}; 
\node[transition](a)at(0,0){$a$};
\node[transition](b)at(4,0){$b$};
\draw(a)[]edge[-latex,bend  left=0]node[above]{$$}(p);
\draw(p)[]edge[-latex,bend left=0]node[above right]{$2$}(b);
\draw(p')[]edge[-latex,bend left=0]node[below]{$$}(a);
\draw(b)[]edge[-latex,bend left=0]node[below]{$$}(p');
\end{tikzpicture}
}
\caption{Convex sets defined by WMG-regions
may be non-totally reachable in $\nsymbol^2$.}
\label{non-reach.fig}
\end{figure}
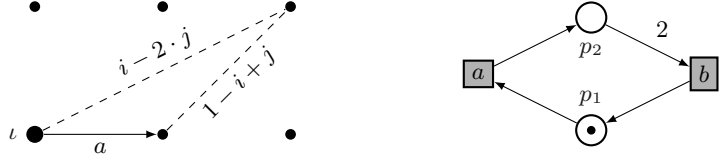

Note that total reachability does not arise from WMG-regions alone, as illustrated 
by Fig.~\ref{non-reach.fig}: on the left,
the points $\is=(0,0)$, $(1,0)$ and $(2,1)$ form a convex set of $\nsymbol^2$, intersection of the WMG-regions
$i-2\cdot j\geq 0$ and $1-i+j\geq 0$ (plus $j\geq 0$ to certify being in $\nsymbol^2$),
but $(2,1)$ is not reachable from $\is$.
These WMG-regions yield the WMG system on the right of the same figure.

A closer look shows 
that state $\is$ corresponds to 
marking $(1,0)$, state $(1,0)$
to marking $(0,1)$
and 
$(2,1)$ to marking $(0,0)$. 
The latter is not reachable, 
but is potentially reachable 
in the sense of \cite{tcs97}.
Let us recall that, 
from the classical state equation $M[\sigma\rangle M'\impl M'=M+C\cdot\Parikh(\sigma)$ 
where $C$ is the incidence matrix, and that a marking $M$ is potentially reachable from the initial marking $M_0$ if 
$M=M_0+C\cdot\alpha$ for some 
$T$-vector $\alpha \ge \zero$ (non-necessarily the Parikh vector of some firing sequence).
Indeed, 
here $C=\left (\begin{array}{rr} -1&1\\1&-2\end{array}\right )$, and $(0,0)=(1,0)+C\cdot(2,1)$ 
(caution: here the vectors are to be considered as column vectors).

Another possible interpretation is to consider the net on the right of Fig.~\ref{non-reach.fig} 
as a {\em continuous} or {\em fluid} one, in the sense of \cite{da10}. 
In those models, 
a transition may be executed fractionally and reachable markings may be real 
vectors with no negative component.   
Thus,
in our case,
we can have the firing sequence
{\small $$(1,0)[a\rangle(0,1)[b^{1/2}\rangle(1/2,0)[a^{1/2}\rangle(0,1/2)[b^{1/4}\rangle(1/4,0)[a^{1/4}\rangle(0,1/4)[b^{1/8}\rangle(1/8,0)\ldots$$}
We cannot finitely reach the marking $(0,0)$, 
but if we allow {\em limit-reachability}, 
then the accumulated firings $2\cdot a+b$ 
finally lead to the marking $(0,0)$.
More generally,  
the whole interior of the shown convex set 
becomes reachable.

Next,
we generalise 
these notions and results 
to any number of labels.  

\subsection{Geometric Characterisation for any Number of Labels}\label{acyck.sct}

Let us consider an acyclic LTS satisfying property \text{\bf b} with $n$ labels $t_1,t_2,\ldots,t_n$.
Again, 
we identify each state $s$ to its distance $\Delta_s\in\nsymbol^n$, 
giving for each $i$ the number of $t_i$'s in any path from $\is$ to $s$. 
Arcs are defined by the relations $s[t_i\rangle s'$ when $s,s'\in S$,
$\Delta_{s'}(t_i)=\Delta_{s}(t_i)+1$ 
and $\Delta_{s'}(t_j)=\Delta_{s}(t_j)$ 
for some $i$ and any $j\neq i$.

We consider special WMG-regions 
of the kind $k+h\cdot x_i-l\cdot x_j\geq 0$ 
for some $k,h\geq 0$, $l>0$ 
and 
$i\neq j$.
In particular, 
each of them is either
parallel to a plane including two axes (if $h>0$), or perpendicular to one axis (if $h=0$). 
From the specialisation of~\cite{ErofeevW17}, 
we deduce the following.

\begin{theorem}[WMG-solvable acyclic  n-ary systems]
\label{acycn.thm}
An acyclic LTS satisfying property \text{\bf b} 
with $n$ different labels
is WMG-solvable if and only if, 
when applied on $\nsymbol^n$, 
its set of states $S$ is 
connected, convex and delimited by 
(i.e., it is the intersection of) a finite set 
of WMG-regions.
A possible solution is then provided by the places corresponding to these regions.
\end{theorem}

However, 
this characterisation is less intuitively
(visually) interpretable when $n>2$. 
Hence it will usually be more efficient to use the general WMG synthesis procedure described in \cite{DH2018}.

\section{A Sufficient Condition of Circular WMG-solvability for any Number of Labels}\label{klabels-cyc.sec} 

In this section,
we provide a general sufficient condition 
for the cyclic solvability of $k$-ary words,
for any positive integer $k$.
This condition,
embodied by the next theorem,
 uses binary subwords obtained by projection\footnote{The projection of a word 
$w\in A^*$  
on a set $A' \subseteq A$ of labels 
is the maximum subword of $w$ 
whose labels belong to $A'$, 
noted 
$\projection{w}{A'}$.
For example, 
the projection of 
the word 
$w = \ell_1 \, \ell_2 \, \ell_3 \, \ell_2$
on the set $\{\ell_1 ,\, \ell_2\}$ is the word 
$\ell_1 \, \ell_2 \, \ell_2$.} 
and containing occurrences of two different labels 
that are contiguous somewhere in the $k$-ary word.
The other binary subwords are not needed since they lack this contiguity and do not capture the direct causality.

\begin{theorem}\label{kary.theo}
Consider any word $w$ over
any finite alphabet $T$ such that $\Parikh(w)$ is prime.
Suppose the following:
$\forall u = \projection{w}{t_1t_2}$ (i.e., the projection of $w$ on $\{t_1,t_2\}$)
for some $t_1,t_2$ such that $t_1 \neq t_2 \in T$,
and $w = (w_1 t_1 t_2 w_2)$ or $w = (t_2 w_3 t_1)$,
$u = v^\ell$ for some positive integer $\ell$,
$\Parikh(v)$ is prime,
and $v$ is cyclically solvable by a circuit.
Then,
$w$ is cyclically solvable with a WMG.
\end{theorem}

\begin{proof}
For every such pair $(t_i,t_j)$, $i< j$, 
let $C_{i,j}= ((P_{i,j},T_{i,j},W_{i,j}),M_{i,j})$ be a circuit solution of $v$ 
for the subword $v^l = u_{i,j} = \projection{w}{t_it_j}$, 
obtained as in the construction of Theorem~\ref{circularbinarywords.thm}.
Assuming all these nets are place-disjoint 
(which is always possible since the Petri net solutions are considered up to isomorphism), 
consider the transition-merging\footnote{Also called sometimes the synchronisation on transitions.} 
of all these marked circuits.
The result is a WMG $\mathcal S'=(N',M_0')$ 
such that
$N'=(P',T,W')$ with 
$P' = \cup_{i,j} P_{i,j}$, 
$T = \cup_{i,j} T_{i,j}$,
$W' = \cup_{i,j} W_{i,j}$,
and $M'_0 = \cup_{i,j} M_{i,j}$.

Let $w$ be of the form $aw'$.
We prove that $a$ is the only transition 
enabled in $\mathcal S'$.

All the subwords of the form $\projection{w}{a,t}$ necessarily start with $a$. 
All the input places of the transition $a$ 
belong to the binary circuits defined 
by these subwords. 
Since these subwords are solvable by marked circuits
which we merged together,
all the input places of $a$ are initially enabled.
Now, let us suppose
that another transition $d$ is also initially enabled in $\mathcal S'$.
Since $d$ is not the first label of $w$,
another label $q$ appears in $w$ just before the first occurrence of $d$.
In the solution of $\projection{w}{d,q}$, $d$ is not initially enabled since $q$ must occur before; 
hence it is not enabled in the merging either. 
We deduce that $a$ is the only transition that is enabled in $S'$.

Now,
the same arguments apply to $w'' = w'a$
whose relevant subwords are solvable by the circuits in the same way,
and we deduce that  the WMG $\mathcal S'$ has the language $\pref(w^*)$.

Note that we did not use explicitely above the special form of $u$. 
Simply, the latter is necessary to build a circuit system $C_{i,j}$ with the language $\pref(u^*)=\pref(v^*)$. 
$C_{i,j}$ is a circular solution for $v$, but not for $u$ unless $\ell=1$.
The fact that the merging $\mathcal S'$ of all the $C_{i,j}$'s yields not only a system with the adequate language $\pref(w^*)$
but a circular solution of $w$ arises from the fact that $\Parikh(w)$ is prime (by Proposition~\ref{WMG.prop}).
We thus deduce that the WMG $\mathcal S'$ solves $w$ cyclically.
\qed
\end{proof}

\section{Synthesis of WMGs from Live Ternary LTS}\label{threelabels.sec} 

In this section,
we provide several conditions of WMG-solvability for a ternary LTS.
We first develop a characterisation of 
WMG-solvability for a subclass of the cyclic ternary words in Subsection~\ref{subclass-ternary.subsec}.
Then,
in Subsection~\ref{ternary-ce.subsec},
we construct two counter-examples to 
this condition:
one for four labels with three different values
in the Parikh vector,
and another one for five labels with 
only two different values.

\subsection{WMG-solvability in a Subclass of the Finite Circular Ternary LTS}\label{subclass-ternary.subsec} 

First,
we prove the other direction of Theorem \ref{kary.theo},
leading to a full characterisation of WMG-solvability
 for a special subclass of the ternary cyclic words.
 
 The proof exploits a WMG with $3$ transitions and $6$ places,
connecting $2$ places to each pair of transitions,
as illustrated in Fig. \ref{WMG3labels}.
In some cases,
a smaller number of places can solve the same LTS,
but we do not aim here at minimising the number of nodes in a solution.

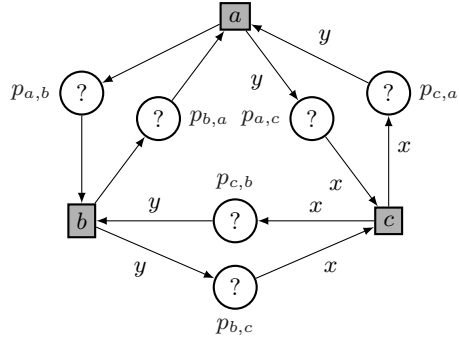
\begin{figure}[htbp]
\begin{center}
\begin{tikzpicture}[scale=0.68]
\node[place,tokens=0](pab)at(-1,1.5)[label=left:$p_{a,b}$]{?};
\node[place,tokens=0](pba)at(0.5,1)[label=right:$p_{b,a}$]{?};
\node[place,tokens=0](pac)at(3.5,1)[label=left:$p_{a,c}$]{?};
\node[place,tokens=0](pca)at(5,1.5)[label=right:$p_{c,a}$]{?};
\node[place,tokens=0](pbc)at(2,-2.3)[label=below:$p_{b,c}$]{?};
\node[place,tokens=0](pcb)at(2,-1)[label=above:$p_{c,b}$]{?};
\node[transition](a)at(2,3){$a$};
\node[transition](b)at(-1,-1){$b$};
\node[transition](c)at(5,-1){$c$};
\draw(pab)[]edge[-latex,bend left=0]node[above left]{}(b);
\draw(pba)[]edge[-latex,bend left=0]node[above right]{}(a);
\draw(a)[]edge[-latex,bend left=0]node[below right]{}(pab);
\draw(b)[]edge[-latex,bend left=0]node[below left]{}(pba);
\draw(pac)[]edge[-latex,bend left=0]node[below left]{$x$}(c);
\draw(c)[]edge[-latex,bend left=0]node[above right]{$x$}(pca);
\draw(pca)[]edge[-latex,bend left=0]node[above right]{$y$}(a);
\draw(a)[]edge[-latex,bend left=0]node[below left]{$y$}(pac);
\draw(pcb)[]edge[-latex,bend left=0]node[above]{$y$}(b);
\draw(b)[]edge[-latex,bend left=0]node[below left]{$y$}(pbc);
\draw(pbc)[]edge[-latex,bend left=0]node[below right]{$x$}(c);
\draw(c)[]edge[-latex,bend left=0]node[above]{$x$}(pcb);
\end{tikzpicture}
\end{center}
\caption{A generic WMG with three labels, 
with minimal T-semiflow $(x,x,y)$ 
and $\gcd(x,y)=1$.
}
\label{WMG3labels}
\end{figure}

\begin{theorem}[Cyclic solvability of ternary words]\label{cyclicsolvternary.theo}
Consider a ternary word $w$ over the alphabet $T$
with Parikh vector $(x,x,y)$ such that $gcd(x,y)=1$.
Then,
$w$ is cyclically solvable with a WMG
if and only if
$\forall u = \projection{w}{t_1t_2}$ 
such that $t_1\neq t_2 \in T$,
and $w = (w_1 t_1 t_2 w_2)$ or $w = (t_2 w_3 t_1)$,
$u = v^\ell$ for some positive integer $\ell$,
$\Parikh(v)$ is prime,
and $v$ is cyclically solvable by a circuit (i.e. a circular net).
\end{theorem}

\begin{proof}
The right-to-left direction of the equivalence, 
assuming the properties on the projections,
is true by Theorem \ref{kary.theo},
for the particular case that $|T|=3$. 
We thus deduce the cyclic solvability.

In the rest of this proof,
we consider the other direction,
assuming circular solvability. 
If $x=y=1$, the claim is trivially obtained since $w=t_1t_2t_3$, up to some permutation, 
and an easy marked graph solution may be found. 
Let us thus assume that $x\neq y$. 

Let us write $T=\{a,b,c\}$.
The general form of a solution has 3 transitions and 6 places (one for each ordered pair of transitions).
Additional places are never necessary in the presence of a T-semiflow. 
Indeed, 
let $p_{u,v}$ be a place between transitions $u$ and $v$, 
$W_u$ the weight on the arc to this place and $W_v$ the one from this place.   
Due to the presence of the T-semiflow $\Parikh(w)$, we have 
$\Parikh(w)(u)\cdot W_u=\Parikh(w)(v)\cdot W_v$,
and we may choose $W_u=\Parikh(w)(v)$ as well as $W_v=\Parikh(w)(u)$.
We may also divide the weights around each place by their gcd.
In our case, this leads to the configuration illustrated by Fig.~\ref{WMG3labels}.
We denote by RG the reachability graph of a solution based on this net. 

We show first that the projection $\projection{w}{ab}$ of $w$ 
on $\{a,b\}$ is of the form $(ab)^k$ or $(ba)^k$ for some positive integer $k$.

There is no pattern $aab$ in $w^2$ (which allows to consider sequences on the border of two consecutive $w$'s)
because, if $M_1[a\rangle M_2[a\rangle M_3[b\rangle$, 
$M_1(p_{c,b})=M_2(p_{c,b})=M_3(p_{c,b})\geq y$
and $M_2(p_{a,b})\geq 1$,
which would also allow to perform $b$ after the first $a$ and RG is not circular.

If there is a pattern $aac$ and $M_1[a\rangle M_2[a\rangle M_3[c\rangle M_4$, $M_2(p_{a,c})\geq y$ and
$M_1(p_{b,c})=M_2(p_{b,c})=M_3(p_{b,c})\geq x$,
hence $y<x$ otherwise $M_2$ already enables $c$ and RG is not circular.
Then $M_4(p_{a,b})\geq 2$, $M_4(p_{c,b})\geq x>y$ and $M_4(p_{c,a})\geq x> y$, so that $M_4[b\rangle$; 
hence $M_4(p_{b,a})=0$ since otherwise we also have $M_4[a\rangle$ and RG is not circular.
We thus have $M_4[ba\rangle M_5$ for some marking $M_5$, with $M_5(p_{b,a})=0$,
so that $M_5$ does not enable $a$; 
$M_5$ does not enable $b$ either since otherwise we could also perform $M_4[bb\rangle$ and again RG is not circular.
Thus, we have $M_4[bac\rangle$, and then we are in a situation similar to the one after the first $c$.
As a consequence, we must have a sequence $M_1[aa(cba)^\omega \rangle$,
and RG is not circular.

Hence, in $w^2$ we cannot have a sequence $aa$, nor $bb$ by symmetry. 

Let us now assume that a pattern $ac^ka$ exists in $w^2$ for some $k\geq 1$.
Since the first firing of $a$ puts a token in $p_{a,b}$ 
and the next firing of $c$ does not enable $b$, we must have $x < y$.
Let us assume in the circular RG that $M_1[ac^ka\rangle M_2[\sigma\rangle M_1$.
$\sigma$ is not empty since it must contain $x$ times $b$.
It cannot end with an $a$, 
since otherwise we have a sequence $aa$, which we already excluded.
It cannot end with a $b$ either, 
since $M_1(p_{b,a})\geq 2$ (in order to fire $a$ twice without a $b$ in between),
so that if $M_3[b\rangle M_1[a\rangle$, $M_3(p_{b,a})\geq 1$, 
we must also have $M_3[a\rangle$, and RG is not circular. 
Hence, 
$\sigma$ ends with a $c$ and for some reachable 
markings $M'_2$ and $M_3$ we have 
 $M_3[c\rangle M_1[ac^k\rangle M'_2[a\rangle M_2$.

We deduce that $M'_2(p_{a,b})\geq 1$, $M'_2(p_{c,b})\geq (k+1)\cdot x$; 
hence $(k+1)\cdot x<y$ otherwise $M'_2$ also enables $b$ and RG is not circular.
Also, $M_3(p_{b,a})\geq 2$ and $M_3(p_{c,a})\geq 2\cdot y - (k+1)\cdot x>y$, 
so that $M_3$ also enables $a$ and again RG is not circular.
As a consequence, 
we cannot have a pattern $ac^ka$, nor $bc^kb$ by symmetry, and 
$\projection{w}{ab}=(ab)^k$ 
or $\projection{w}{ab}=(ba)^k$ for the positive integer $k=x$.
With $v=ab$ or $v=ba$, we have the adequate solvability property for $\projection{w}{ab}$, 
and we can assume in the following that the sum of the tokens present in places $p_{a,b}$ and $p_{b,a}$
is $1$ for all reachable markings.

Let us now suppose that we have a WMG $\mathcal S$ 
solving $w$ cyclically, whose underlying net is pictured in Fig.~\ref{WMG3labels}.
From the previous results, 
we can assume that  $M_0(p_{a,b})+M_0(p_{b,a})=1$ in $\mathcal S$,
 this equality being preserved by all reachable markings. 
To show that $u= \projection{w}{ac}$ has the adequate form (the case for $\projection{w}{bc}$ is symmetrical),
let us consider the circuit $C_{ac}$, restriction of $\mathcal S$ to $p_{a,c},c,p_{c,a},a$.

Let us assume in the following that 
$u$ cannot be written under the form $u=v^\ell$ for some positive integer $\ell$, 
where $\Parikh(v)$ is prime and $v$ is cyclically solvable. 
Since $gcd(x,y)=gcd(\Parikh(w)(a),\Parikh(w)(c))=gcd(\Parikh(u)(a),\Parikh(u)(c))=1$,
$\Parikh(u)$ is prime and $u=v$ with $\ell=1$,
hence $u$ is not cyclically solvable.
For the net $N$ considered,
this implies the existence of some prefix $\sigma_{ac}$ of $u$ 
such that,
for every initial marking of $C_{ac}$ that enables the sequence $u$ in this circuit,
the marking reached by firing $\sigma_{ac}$
necessarily enables both places $p_{a,c}$ and $p_{c,a}$.
Indeed,
Theorem~\ref{circularbinarywords.thm} specifies the finite set of all possible minimal markings
that allow cyclic solvability,
and each such marking enables exactly one place of the circuit.
Every other non-circular reachability graph is defined by some larger initial marking
 and contains a marking that enables both places.\\
Thus,
for any initial marking $M_0$ that makes the system $\mathcal S = (N,M_0)$ solve $w$ cyclically,
the smallest prefix of $w$ whose projection on $\{a,c\}$ equals $\sigma_{ac}$ 
leads to a marking $M$ in the WMG that enables $p_{a,c}$ and $p_{c,a}$.

Hereafter, 
we consider all the cases in which
either $a$ or $c$ is enabled from $M$.
In each case,
we describe the shape of the LTS
and
deduce from it a reachable marking
that enables two transitions, 
hence a contradiction.

Case $x>y$: In this case, in $\mathcal S$, we cannot have two consecutive $c$'s.

$-$ Subcase in which $M$ enables the place $p_{a,c}$ as well as the transition $a$ in the WMG, 
hence its input places $p_{b,a}$ and $p_{c,a}$.
Since $M[a\rangle$, transition $c$ is not enabled at $M$, implying that $p_{b,c}$ is not enabled by $M$.
We deduce: $M(p_{a,c}) > M(p_{b,c})$.
Since $p_{a,c}$ is enabled by $M$,
the last occurrence of a transition before the next firing of $c$ is necessarily $b$,
implying: $M[(ab)^kc\rangle M_1$ for some integer $k \ge 1$  and some marking $M_1$.
The inequality mentioned above is still valid at $M_1$,
i.e. 
$M_1(p_{a,c}) > M_1(p_{b,c})$,
and we iterate the same arguments from $M_1$ to deduce
that the rotation $w_M$ of $w$ starting at $M$ 
is of the form $(ab)^{k_1}c \ldots (ab)^{k_y}c$ with $\sum_{i=1..y} k_i=x$ and each $k_i$ is positive.

$-$ Subcase in which $M$ enables the place $p_{c,a}$ as well as the transition $c$ in the WMG, 
hence its input places $p_{a,c}$ and $p_{b,c}$.
Thus,
the firing of $c$ from $M$
cannot enable $a$,
implying that $M(p_{c,b})<M(p_{c,a})$
and that $M[c(ba)^kc\rangle M_1$ for some positive integer $k$ and a marking $M_1$.
The inequality is still valid at $M_1$,
i.e. $M_1(p_{c,b})<M_1(p_{c,a})$,
from which
we deduce  
that the rotation $w_M$ of $w$ starting at $M$ 
is of the form $c(ba)^{k_1} \ldots c(ba)^{k_y}$ with $\sum_{i=1..y} k_i=x$  and each $k_i$ is positive.

Case $x \le y$:

$-$ Subcase in which $M$ enables the place $p_{a,c}$ as well as the transition $a$ in the WMG, 
hence its input places $p_{c,a}$ and $p_{b,a}$.
Thus,
the firing of $a$ from $M$
cannot enable $c$,
implying that $M(p_{b,c})<M(p_{a,c})$
and that $M[abc^k\rangle M_1$ for some positive integer $k$ and a marking $M_1$,
at which the same inequality is still valid.
We deduce  
that the rotation $w_M$ of $w$ starting at $M$ is of the form $abc^{k_1} \ldots abc^{k_x}$ with $\sum_{i=1..x} k_i=y$
 and each $k_i$ is positive.

$-$ Subcase in which $M$ enables the place $p_{c,a}$ as well as the transition $c$ in the WMG, 
hence its input places $p_{a,c}$ and $p_{b,c}$.
Thus, 
firing one or several $c$'s from $M$ does not enable $a$,
and 
$M(p_{c,b})<M(p_{c,a})$,
implying that 
$M[c^kba\rangle M_1$ for some positive integer $k$ and a marking $M_1$,
at which the same inequality is still valid.
We deduce 
that the rotation $w_M$ of $w$ starting at $M$ is of the form $c^{k_1}ba \ldots c^{k_x}ba$ with $\sum_{i=1..x} k_i=y$
 and each $k_i$ is positive.

In each of the four cases developed above,
we observe that each sequence of $ab$ or $ba$ could be seen as an atomic firing,
and $\projection{w_M}{b,c}$ is obtained from $\projection{w_M}{a,c}$ by renaming each $a$ into one $b$.
This implies that the deletion of the initial useless tokens (also known as frozen tokens, i.e. never used by any firing)
yields a system in which some reachable marking
distributes the tokens in the same way in the places between $c$ and $a$ as in the places between $c$ and $b$.
This is for example the case of the marking M if it does not contain useless tokens.\\
We deduce that $M$ (with or without useless tokens) enables 
all four places $p_{a,c}$, $p_{c,a}$, $p_{b,c}$ and $p_{c,b}$,
thus enabling two transitions of the WMG at least.
This contradicts the cyclic solvability of $w$,
implying that $v=u$ is cyclically solvable by a circuit.
Hence the claim.
\qed \end{proof}

\subsection{Counter-examples for $4$ and $5$ Labels}\label{ternary-ce.subsec} 

In Theorem \ref{cyclicsolvternary.theo},
we provided a characterisation of cyclic WMG-solvability 
for ternary words $w$ such that $\Parikh(w)$ is prime with two values.
However, 
this result does not apply to words $w$ 
over $4$ labels with $3$ values
nor $5$ labels with $2$ values,
even if $\Parikh(w)$ is prime.
Indeed, 
Fig. \ref{counterex.fig} 
pictures two counter-examples:
on the left,
the WMG cyclically solves the word 
$w = aacbbdabd$ with $\Parikh(w)=(3,3,1,2)$, 
which is prime,
while its projection $u=aabbab$ on $\{a,b\}$ 
leads to $v=u$, 
and $\Parikh(v)=(3,3)$ is not prime, 
hence is not cyclically solvable by a WMG;
on the right,
the WMG cyclically solves the word 
$w = aacbbeabd$ with $\Parikh(w)=(3,3,1,1,1)$, which is prime,
while its projection $u=aabbab$ on $\{a,b\}$ 
leads to $v=u$, 
and $\Parikh(v)=(3,3)$ 
is not cyclically solvable by a WMG.

\begin{figure}[htbp]
\begin{center}
\begin{tikzpicture}[scale=0.8]
\begin{scope}[]
\node[place,tokens=0](p0)at(2,2.5)[]{$ $};
\node[place,tokens=1](p1)at(0,3.5)[]{$ $};
\node[place,tokens=4](p2)at(0,1.5)[]{$ $};
\node[place,tokens=0](p3)at(4,1.5)[]{$ $};
\node[place,tokens=0](p4)at(4,3.5)[]{$ $};
\node[transition](a)at(0,2.5){$a$}
	edge[post]node[left]{$ $}(p1)
	edge[post]node[above]{$ $}(p0)
	edge[pre]node[right, near end]{$2$}(p2)
	;
\node[transition](b)at(4,2.5){$b$}
	edge[pre]node[above]{$ $}(p0)
	edge[post]node[left, near end]{$2$}(p3)
	edge[pre]node[below right]{$ $}(p4)
	;
\node[transition](c)at(2,3.5){$c$}
	edge[pre]node[below]{$3$}(p1)
	edge[post]node[below]{$3$}(p4)
	;
\node[transition](d)at(2,1.5){$d$}
	edge[pre]node[above]{$3$}(p3)
	edge[post]node[above]{$3$}(p2)
	;
\end{scope}
\end{tikzpicture}
\hspace{1cm}
\begin{tikzpicture}[scale=0.8]
\begin{scope}[]
\node[place,tokens=1](p0)at(5,0)[]{$ $};
\node[place,tokens=0](p1)at(5,3)[]{$ $};
\node[place,tokens=2](p2)at(0,0)[]{$ $};
\node[place,tokens=1](p3)at(0,3)[]{$ $};
\node[place,tokens=0](p4)at(2.5,2)[]{$ $};
\node[place,tokens=3](p5)at(1,1)[]{$ $};
\node[place,tokens=0](p6)at(4,1)[]{$ $};
\node[transition](a)at(0,2){$a$}
	edge[post]node[above]{$ $}(p4)
	edge[post]node[above]{$ $}(p3)
	edge[pre]node[above]{$ $}(p5)
	edge[pre]node[left]{$ $}(p2)
	;
\node[transition](b)at(5,2){$b$}
	edge[pre]node[]{$ $}(p1)
	edge[pre]node[above]{$ $}(p4)
	edge[post]node[above, near end]{$ $}(p6)
	edge[post]node[]{$ $}(p0)
	;
\node[transition](c)at(2.5,3){$c$}
	edge[pre]node[below]{$3$}(p3)
	edge[post]node[below]{$3$}(p1)
	;
\node[transition](d)at(2.5,1){$d$}
	edge[pre]node[above]{$3$}(p6)
	edge[post]node[above]{$3$}(p5)
	;
\node[transition](e)at(2.5,0){$e$}
	edge[post]node[above]{$3$}(p2)
	edge[pre]node[above]{$3$}(p0)
	;
\end{scope}
\end{tikzpicture}
\caption{The WMG on the left solves $aacbbdabd$
cyclically,
and the WMG on the right solves
$aacbbeabd$ cyclically.}
\label{counterex.fig}
\end{center}
\end{figure}
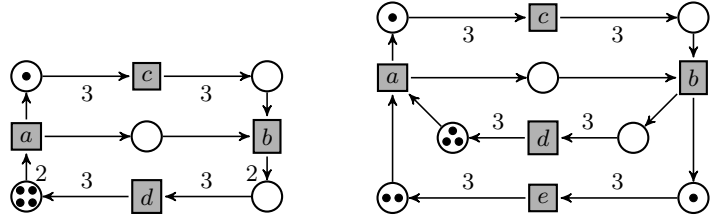

However, presently we do not know what happens for ternary words $w$ such that $\Parikh(w)$ is prime with three values,
nor when $w$ has four labels and $\Parikh(w)$ is prime with two values. 

\section{Conclusions and Perspectives}\label{conclu.sec}

In this work, we specialised 
previous methods dedicated to the analysis and synthesis of weighted marked graphs,
a well-known and useful subclass of weighted Petri nets
allowing to model various real-world applications.

By restricting the size of the alphabet to $2$ labels,
we provided a characterisation of the WMG-solvable labelled transition systems formed of a single cycle.
We also extended this investigation
to finite LTS containing several cycles,
and
to infinite LTS.

Then,
leaving out the restriction
on the number of labels,  
we developed a geometric characterisation
for acyclic LTS,
using convex sets and the theory
of regions;
in the case circular LTS,  
we proposed a sufficient condition of WMG-solvability. 

We exploited this sufficient condition  
to obtain
a full characterisation of circular 
WMG-solvability 
for a subset of the possible Parikh vectors
over three labels.

Finally, 
we proved that this condition for $3$ labels 
does not extend to circular LTSs with 
$4$ labels 
and three different Parikh values,
nor with $5$ labels and two Parikh values.  

As perspectives of this work,   
we believe that relaxations
of our statements
may lead to other characterisations
of WMG-solvable LTS,
together with efficient algorithms
for their analysis and synthesis.

\section*{Acknowledgements} 
We would like to thank the anonymous referees for their involvement and useful suggestions.

\bibliographystyle{splncs}
\bibliography{topnoc-2019-DEH}

\begin{thebibliography}{10}
\providecommand{\url}[1]{\texttt{#1}}
\providecommand{\urlprefix}{URL }
\providecommand{\doi}[1]{https://doi.org/#1}

\bibitem{bbd}
Badouel, E., Bernardinello, L., Darondeau, P.: Petri Net Synthesis.
  Springer-Verlag (2015), \url{https://doi.org/10.1007/978-3-662-47967-4}

\bibitem{BarylskaBEMP15}
Barylska, K., Best, E., Erofeev, E., Mikulski, L., Piatkowski, M.: On binary
  words being {P}etri net solvable. In: Proceedings of the International
  Workshop on Algorithms {\&} Theories for the Analysis of Event Data, {ATAED}
  2015, Brussels, Belgium. pp. 1--15 (2015),
  \url{http://ceur-ws.org/Vol-1371/paper01.pdf}

\bibitem{BarylskaBEMP16}
Barylska, K., Best, E., Erofeev, E., Mikulski, L., Piatkowski, M.: Conditions
  for {P}etri net solvable binary words. T. Petri Nets and Other Models of
  Concurrency  \textbf{11},  137--159 (2016).
  \doi{10.1007/978-3-662-53401-4\_7}

\bibitem{besdev-ccc25}
Best, E., Devillers, R.: Synthesis and reengineering of persistent systems.
  Acta Inf.  \textbf{52}(1),  35--60 (2015). \doi{10.1007/s00236-014-0209-7},
  \url{http://dx.doi.org/10.1007/s00236-014-0209-7}

\bibitem{binary16}
Best, E., Erofeev, E., Schlachter, U., Wimmel, H.: Characterising petri net
  solvable binary words. In: Kordon, F., Moldt, D. (eds.) Application and
  Theory of Petri Nets and Concurrency. pp. 39--58. Springer International
  Publishing, Cham (2016)

\bibitem{TCS17}
Best, E., Hujsa, T., Wimmel, H.: Sufficient conditions for the marked graph
  realisability of labelled transition systems. Theoretical Computer Science
  (2017). \doi{https://doi.org/10.1016/j.tcs.2017.10.006},
  \url{http://www.sciencedirect.com/science/article/pii/S0304397517307181}

\bibitem{chep}
Commoner, F., Holt, A., Even, S., Pnueli, A.: Marked directed graphs. J.
  Comput. Syst. Sci  \textbf{5}(5),  511--523 (1971),
  \url{https://doi.org/10.1016/S0022-0000(71)80013-2}

\bibitem{crespi-mandrioli-75}
Crespi{-}Reghizzi, S., Mandrioli, D.: A decidability theorem for a class of
  vector-addition systems. Inf. Process. Lett.  \textbf{3}(3),  78--80 (1975).
  \doi{10.1016/0020-0190(75)90020-4},
  \url{http://dx.doi.org/10.1016/0020-0190(75)90020-4}

\bibitem{da10}
David, R., Alla, H.: Discrete, Continuous, and Hybrid Petri Nets. Springer
  Publishing Company, Incorporated, 2nd edn. (2010).
  \doi{10.1007/978-3-642-10669-9}

\bibitem{ACSD13}
Delosme, J.M., Hujsa, T., Munier-Kordon, A.: Polynomial sufficient conditions
  of well-behavedness for weighted join-free and choice-free systems. In: 13th
  International Conference on Application of Concurrency to System Design. pp.
  90--99 (July 2013). \doi{10.1109/ACSD.2013.12}

\bibitem{DesEsp}
Desel, J., Esparza, J.: {F}ree {C}hoice {P}etri {N}ets, Cambridge Tracts in
  Theoretical Computer Science, vol.~40. Cambridge University Press, New York,
  USA (1995)

\bibitem{dev-ACSD16}
Devillers, R.: {Products of Transition Systems and Additions of {Petri} Nets}.
  In: Proc. 16th International Conference on Application of Concurrency to
  System Design ({ACSD} 2016) J. Desel and A. Yakovlev (eds). pp. 65--73
  (2016), \url{https://doi.org/10.1109/ACSD.2016.10}

\bibitem{devacta17}
Devillers, R.: Factorisation of transition systems. Acta Informatica  (2017),
  \url{https://doi.org/10.1007/s00236-017-0300-y}

\bibitem{DH2018}
Devillers, R., Hujsa, T.: Analysis and synthesis of weighted marked graph
  {P}etri nets. In: Khomenko, V., Roux, O.H. (eds.) Application and Theory of
  {P}etri Nets and Concurrency: 39th International Conference, PETRI NETS 2018,
  Bratislava, Slovakia, 2018, Proceedings. pp. 19--39. Springer International
  Publishing (2018)

\bibitem{AtaedDEH18}
Devillers, R.R., Erofeev, E., Hujsa, T.: Synthesis of weighted marked graphs
  from constrained labelled transition systems. In: Proceedings of the
  International Workshop on Algorithms {\&} Theories for the Analysis of Event
  Data, Bratislava, Slovakia. pp. 75--90 (2018),
  \url{http://ceur-ws.org/Vol-2115/ATAED2018-75-90.pdf}

\bibitem{doignon73}
Doignon, J.P.: Convexity in cristallographical lattices. Journal of Geometry
  \textbf{3(1)},  71--85 (1973). \doi{10.1007/BF01949705}

\bibitem{ErofeevBMP16}
Erofeev, E., Barylska, K., Mikulski, L., Piatkowski, M.: Generating all minimal
  {P}etri net unsolvable binary words. In: Proceedings of the Prague
  Stringology Conference 2016, Prague, Czech Republic. pp. 33--46 (2016),
  \url{http://www.stringology.org/event/2016/p04.html}

\bibitem{ErofeevW17}
Erofeev, E., Wimmel, H.: Reachability graphs of two-transition {P}etri nets.
  In: Proceedings of the International Workshop on Algorithms {\&} Theories for
  the Analysis of Event Data 2017, Zaragoza, Spain. pp. 39--54 (2017),
  \url{http://ceur-ws.org/Vol-1847/paper03.pdf}

\bibitem{phdhujsa}
Hujsa, T.: Contribution to the study of weighted {P}etri nets. Ph.D. thesis,
  Pierre and {M}arie {C}urie University, {P}aris, {F}rance (2014),
  \url{https://tel.archives-ouvertes.fr/tel-01127406}

\bibitem{PN14}
Hujsa, T., Delosme, J.M., Munier-Kordon, A.: On the reversibility of
  well-behaved weighted choice-free systems. In: Ciardo, G., Kindler, E. (eds.)
  Application and Theory of Petri Nets and Concurrency. pp. 334--353. Springer
  International Publishing, Cham (2014)

\bibitem{TECS14}
Hujsa, T., Delosme, J.M., Munier-Kordon, A.: Polynomial sufficient conditions
  of well-behavedness and home markings in subclasses of weighted {P}etri nets.
  ACM Trans. Embed. Comput. Syst.  \textbf{13}(4s),  141:1--141:25 (Jul 2014).
  \doi{10.1145/2627349}, \url{http://doi.acm.org/10.1145/2627349}

\bibitem{Hujsa2016}
Hujsa, T., Delosme, J.M., Munier-Kordon, A.: On liveness and reversibility of
  equal-conflict {P}etri nets. Fundamenta Informaticae  \textbf{146}(1),
  83--119 (2016), \url{https://doi.org/10.3233/FI-2016-1376}

\bibitem{HD2017}
Hujsa, T., Devillers, R.: On liveness and deadlockability in subclasses of
  weighted {P}etri nets. In: van~der Aalst, W., Best, E. (eds.) Application and
  Theory of {P}etri Nets and Concurrency: 38th International Conference, PETRI
  NETS 2017, Zaragoza, Spain, June 25--30, 2017, Proceedings. pp. 267--287.
  Springer International Publishing, Cham (2017).
  \doi{10.1007/978-3-319-57861-3\_16},
  \url{https://doi.org/10.1007/978-3-319-57861-3\_16}

\bibitem{keller}
Keller, R.M.: {A Fundamental Theorem of Asynchronous Parallel Computation}. In:
  Sagamore Computer Conference, August 20-23 1974, LNCS Vol. 24. pp. 102--112
  (1975), \url{https://doi.org/10.1007/3-540-07135-0\_113}

\bibitem{murata89}
Murata, T.: Petri nets: properties, analysis and applications. Proceedings of
  the IEEE  \textbf{77}(4),  541--580 (April 1989)

\bibitem{WTS92}
Teruel, E., Chrzastowski-Wachtel, P., Colom, J.M., Silva, M.: On weighted
  {T}-systems. In: Jensen, K. (ed.) 13th International Conference on
  Application and Theory of {Petri} Nets and Concurrency ({ICATPN}), LNCS.
  vol.~616, pp. 348--367. Springer, Berlin, Heidelberg (1992),
  \url{https://doi.org//10.1007/3-540-55676-1\_20}

\bibitem{tcs97}
Teruel, E., Colom, J.M., Silva, M.: {Choice-Free {Petri} Nets: a Model for
  Deterministic Concurrent Systems with Bulk Services and Arrivals}. {IEEE}
  Transactions on Systems, Man, and Cybernetics, Part {A}  \textbf{27}(1),
  73--83 (1997). \doi{10.1109/3468.553226},
  \url{http://dx.doi.org/10.1109/3468.553226}

\bibitem{STECS}
Teruel, E., Silva, M.: Structure theory of {E}qual {C}onflict systems.
  Theoretical Computer Science  \textbf{153}(1{\&}2),  271--300 (1996),
  \url{https://doi.org/10.1016/0304-3975(95)00124-7}

\end{thebibliography}

\end{document}